\documentclass[10pt, a4paper]{article}

\usepackage{amsmath,amssymb,graphicx,epsfig}
\usepackage[tight,footnotesize]{subfigure}
\usepackage{textcomp}
\usepackage[ruled,vlined]{algorithm2e}
\usepackage[active]{srcltx} 
\usepackage{subfigure}

\newtheorem{theorem}{Theorem}[section]
\newtheorem{lemma}[theorem]{Lemma}
\newtheorem{claim}[theorem]{Claim}
\newtheorem{assumption}[theorem]{Assumption}
\newtheorem{definition}[theorem]{Definition}
\newtheorem{example}{Example}

\newcommand{\QED}{\hfill~\rule[-1pt]{6pt}{6pt}}
\newenvironment{proof}{{\noindent \bf Proof \ }}{\QED}

\renewcommand{\hat}{\widehat}

\newcommand{\beq}{\begin{equation}}
\newcommand{\eeq}{\end{equation}}

\newcommand{\re}{\hbox{\rm I\kern -.2em R}}
\newcommand{\real}{\hbox{\rm I\kern -.2em R}}
\newcommand{\integer}{\hbox{\rm Z\kern -.4em Z}}
\newcommand{\complex}{\hbox{\rm I\kern -.5em C}}
\newcommand{\id}{\hbox{\rm 1\kern -.25em I}}

\def\R{{\cal R }}
\def\S{{\cal S }}

\def\beqn{\begin{eqnarray}}

\def\eeqn{\end{eqnarray}}

\def\beqnn{\begin{eqnarray*}}

\def\eeqnn{\end{eqnarray*}}

\title{Robustness  Analysis  for  Battery Supported Cyber-Physical Systems}
\author{\small FUMIN ZHANG, ZHENWU SHI, and SHAYOK MUKHOPADHYAY \\
 Georgia Institute of Technology}

\begin{document}
\date{}
\maketitle

\footnotetext[1]{This paper has been accepted by ACM Transactions in Embedded Computing Systems (TECS) in October, 2011.}
\footnotetext[2]{This research was partially supported by the ONR grants N00014-08-1-1007, N00014-09-1-1074, and N00014-10-10712(YIP), and NSF grants ECCS-0841195, ECCS-0845333(CAREER) and CNS-0931576. Author’s address: Fumin Zhang, Zhenwu Shi and Shayok Mukhopadhyay, email: \{fumin, zwshi and shayok\}@gatech.edu.}

\noindent {\bf \Large Abstract}\\

\vspace{-0.5mm}
\noindent This paper establishes a novel analytical approach to quantify robustness of scheduling and battery management for battery supported cyber-physical systems. A dynamic schedulability test is introduced to determine whether tasks are schedulable within a finite time window. The test is used to measure robustness of  a real-time scheduling algorithm by evaluating the strength of computing time perturbations  that break schedulability at runtime. Robustness of battery management is quantified analytically by an adaptive threshold on the state of charge. The adaptive threshold significantly reduces the false alarm rate for battery management algorithms to decide when a battery needs to be replaced.\\

{\small \noindent Categories and Subject Descriptors: C.3 [{\bf Special-Purpose and Application-Based Systems}]: {\it Real-time and embedded systems} ; D.4.1 [{\bf Operating System}]: Process Management--{\it Scheduling}; G.4 [{\bf Mathematical Software}] }\\

\vspace{-3mm}
{\small \noindent General Terms: Algorithms, Design, Performance, Reliability, Management, Theory}\\

\vspace{-3mm}
{\small \noindent {Additional Key Words and Phrases:} {Cyber-physical systems, battery management,  dynamic timing model, dynamic schedulability test} }\\

%
%

%
%
%
%
%

\section{Introduction}
Cyber physical systems (CPS) theory represents a novel research direction aiming to establish foundations for a tight integration of computing and physical processes \cite{ShaCPS09,WolfCPS09,LeeCPS08}.  CPS research unifies domain specific design methods for   subsystems to achieve desirable overall performance of the entire system.
We are interested in battery supported  CPS (CPSb) where control of physical systems and the underlying  computing activities are confined by battery capacity, such as  mobile devices. In CPSb, the battery, the actuators and the sensors can be viewed as physical components, while the embedded computers can be viewed as cyber components.
The cyber and the physical components interact with each other so that no complete understanding can be gained by studying any component alone. The total discharge currents from the battery include currents drawn from all cyber and physical components as  results of the interactions between these components. In order to estimate the remaining capacity of the battery or predict the remaining battery life, knowledge of the interactions among all cyber-physical components  are necessary.

CPSb can be tested and verified using computer simulation tools that simulate all its components. Intensive simulations at the design phase usually  achieve tolerance of perturbations that can be predicted. Prototypes of CPSb can then be verified using experiments.  Exhaustive simulations and experiments are usually labor intensive and costly. Simpler yet less expensive approaches are desirable.

We propose an analytical approach to study CPSb. The analytical approach combines simplified mathematical models that capture the characteristic behaviors of each component of a CPSb. This approach is approximate in its nature. But since all CPSb components are modeled uniformly with mathematical equations, interactions between the CPSb components are naturally described as coupling terms between the mathematical models.  Hence the analytical approach is well suited for gaining insight into the interactions among the CPSb components.  Furthermore,  mathematical insights into CPSb are greatly appreciated when perturbations unpredictable at the design phase may force the systems  to work in conditions that are near or beyond the design envelopes where reliability becomes less guaranteed.

 In this paper, we follow an analytical approach to develop  mathematical tools to measure  robustness of real-time scheduling algorithms and battery management algorithms for CPSb during  runtime.  The mathematical tools produce exact solutions in terms of mathematical formulas to describe the interactions between embedded computers and batteries, which are complementary to results obtained  using simulation or experimental methods. In the rest of the introduction, we  briefly review some background knowledge from literature that is closely related to our work, followed by the research problems addressed and the contributions made by this paper.

 \subsection{Literature Review}
An important branch of  real-time systems research is to study schedulabilty. It tries to ascertain whether a set of real-time tasks can be computed by a processor under proper scheduling.
 The study of utilization based schedulability tests can be traced back to the rate monotonic scheduling (RMS) and earliest deadline first scheduling (EDF) \cite{liulayland73}. It has been shown that if a set of real-time tasks fall below a utilization bound, then they will be schedulabe.  Since then, extensive research has been conducted on periodic tasks to improve the utilization bounds \cite{Lehoczky_1989,AKMok_1991,Bini03} or to relax  assumptions \cite{Lehoczky_arbitrarydeadline,ABurns_offset} that are used to derive these bounds. Some important utilization bounds for non-periodic systems are also derived in \cite{Abdelzaher04}.  Schedulability tests based on utilization bounds are easy to compute. Therefore, they are often used during runtime (online), but are constrained by limited computational power. Schedulability tests based on utilization bounds are typically conservative because they can fail on schedulable task sets. This drawback  leads to exact schedulability tests \cite{Audsley_RTA,Lehoczky_1989,Pandya_RTA}. Some recent advancements have been reported  on exact schedulability tests \cite{Andersson_rtas,Fenxiang} with improved computational efficiency.

Robustness is well studied for feedback control systems and has seen successful applications \cite{kz}.
For real-time scheduling, robustness is introduced as a measure of the tolerance of a scheduling algorithm to variations in computing time e.g. perturbations \cite{Stankovic_MSFreport,Regehr_rtss02,Iainone}.  These works measure robustness by using a scaling factor (greater than one) for computing times that are long enough to cause a loss of schedulability. The robustness measure is computed using the binary search method, which  limits it to non-periodic tasks. Based on this notion of robustness, the method of elastic scheduling \cite{Buttazzo02,HuLemmon06} adjusts the periods of tasks to accommodate runtime perturbations.

Prediction of the state of charge (SoC, or the remaining battery capacity) is a basic function for all battery management algorithms \cite{rvw_model}.  A dynamic nonlinear battery model \cite{rincon} and a particle filter will be used to predict the SoC in this paper.  Different scheduling and control methods result in different ``load profiles" that affect the operational
life of a battery, hence various battery management algorithms are proposed \cite{Rakhmatov2,kim_rtas} to adjust the scheduling and control to prolong battery life.  These previous results usually rely on optimization methods.

\subsection{Research Problems and Contributions}

We provide robustness analysis for CPSb by measuring robustness of  both  real-time scheduling and battery management algorithms. Two types of perturbations are studied in this paper: perturbations to the computing times of real-time tasks, and perturbations to the SoC and parameters of batteries. The perturbations to the computing times may extend or shorten the time spent to compute real-time tasks. The perturbations to the SoC may increase or decrease the SoC. We assume that these perturbations have not been accounted for at the design stage, but have to be tolerated at runtime.
\begin{itemize}
\item {\bf How is robustness measured?} Robustness of a real-time scheduling algorithm is measured as the maximum strength of perturbations on the computing times of scheduled tasks that will not cause loss of schedulability.
 Robustness of a battery management algorithm is measured by its ability to trigger the switching of a used battery out of the system before the SoC of the battery drops below a threshold that indicates instability, even under perturbations to the SoC and battery parameters.

\item {\bf What methods are developed to study robustness of real-time scheduling algorithms?} We first developed a new mathematical model for the scheduled behaviors of real-time tasks.  We then study schedulability of  these tasks within a receding finite time window, and devise a dynamic schedulability test to give sufficient and necessary conditions for schedulability of acyclic task sets (e.g. tasks that are not necessarily periodic) under any priority based scheduling algorithm.
  The maximum strength of the perturbations that will not break schedulability  can then be determined analytically. This tolerable strength of the perturbations provides a measure for  robustness of the scheduling algorithm employed.

\item{\bf What methods are developed to study robustness of battery management algorithms?} The mathematical models of  real-time scheduling  are combined with  the controllers developed in our previous work \cite{ZWS_RTSS08} to generate predictions for the total battery discharge current. This prediction is then used to predict the SoC of batteries analytically at runtime.
 Due to nonlinearities inherent in battery behaviors, we introduce a measure for the robustness of battery management algorithms based on Lyapunov stability criteria \cite{Khalil}.  We then introduce an adaptive battery switching algorithm based on the Lyapunov stability test to determine when used battery should be replaced.

\item{\bf What are the contributions for CPS?} We have developed unified mathematical models for real-time scheduling in embedded computers that form the cyber components of CPSb, and for the discharging of batteries that form the physical components of CPSb. These mathematical models are also integrated with the feedback controller  developed  in our previous work \cite{ZWS_RTSS08}.  By combining these mathematical models, we are able to study the interactions between the cyber and physical components analytically, this is well aligned with the main theme of CPS research. Several benefits have been generated by this analytical approach:
\begin{itemize}
\item  Our robustness analysis incorporates both real-time scheduling and battery management algorithms. These results have not been reported in literature. The robustness measures are able to account for situations at runtime that are unexpected at the design stage.
\item The dynamic schedulability test is an exact schedulability test for non-periodic task sets. We have also generalized the notion of robustness from periodic task sets to non-periodic task sets. These results are novel and complementary in comparison to the literature reviewed.
\item Compared to existing battery management algorithms that use fixed thresholding for output voltage or for SoC \cite{BattHandBook,kim_rtas,philipsbook} to determine when to replace a used battery, our adaptive battery switching algorithm effectively reduces the false alarm rate.
\end{itemize}

\end{itemize}

The paper is organized as follows. Section \ref{section:robustness} discusses robustness of real-time scheduling algorithms. Section \ref{batteryman} studies robustness for battery management algorithms. Section \ref{sec:application} demonstrates the applications of the mathematical tools developed in this paper to a typical
CPSb. Section \ref{con1} provides summary and conclusions.

\section{Robustness of  Real-time Scheduling Algorithms}\label{section:robustness}
A real-time scheduling algorithm assigns priorities to a set of real-time tasks so that
all tasks can be computed on time on a processor.
At the design phase of a  real-time system, the parameters of tasks, such as computing times and deadlines,  are usually determined based on desired performance and experimental data.  We call these parameters the nominal characteristics. During runtime, the actual computing times and deadlines  may deviate from the nominal values due to variations in the software, hardware, and the environment. These deviations are usually considered as online perturbations.
For perturbations that can be predicted at the design phase, such as changes in task modes, the ``design-of-experiments" method may be applied to verify whether a scheduling algorithm can
tolerate such perturbations \cite{Iainone,Iaintwo}. Usually there exist online perturbations that may be difficult to predict at the design stage, such as the transient overload of certain tasks and the arriving of unexpected tasks.  In this section, we introduce mathematical tools to measure tolerance of a real-time scheduling algorithm to online perturbations.

Perturbations occurring online can change timing of the real-time tasks. It can cause a set of schedulable tasks to become unschedulable. Thus it is necessary to introduce a way to evaluate the schedulability during runtime as follows:
\begin{definition}
A \emph{dynamic schedulability test} over a  time interval $[t_a, t_b]$ checks if all task instances are able to meet their deadlines within $[t_a, t_b]$.
\end{definition}
As the starting time $t_a$ increases, the time interval $[t_a, t_b]$ will slide forward.
The length of the interval $(t_b-t_a)$ depends on how confident we are to  predict  the actual characteristics of the  real-time tasks to perform the schedulability test.  All mathematical tools developed in this section are
centered around the dynamic schedulability test within the time interval $[t_a, t_b]$.

\subsection{A Task Model}

 For theoretical rigor, let us define the task set that will be scheduled, which will include both periodic and aperiodic (non periodic) tasks.
We consider a task set $\Gamma$ of $N$ independent  hard real-time tasks $\Gamma=\{\tau_1, \tau_2, \dotsb, \tau_N\}$ running on a single processor. Let $\tau_n$ be any task in $\Gamma$. Each task in $\Gamma$ consists of an infinite sequence of instances.  We use  the notation $\tau_n^{k}$ to represent the $k$-th  instance of task $\tau_n$. The instance $\tau_n^{k}$ is characterized by its time of arrival  $a^{k}_{n}$, its computing time $C^{k}_{n}$ and its relative deadline $T^{k}_{n}$ measured from its time of arrival. The absolute deadline of  $\tau_n^{k}$ is then defined as $a^{k}_{n}+T^{k}_{n}$.

For theoretical rigor,  we make all tasks in the task set $\Gamma$ {\it acyclic} (\cite{Abdelzaher04}) as defined befow:
\begin{definition} \label{definition:acyclictask}
A task $\tau_n$ is  \emph{acyclic}  if and only if $\tau_n$ satisfies the following properties:
\begin{enumerate}
\item different instances of $\tau_n$ are allowed to have different computing times and different relative deadlines, as long as $0 \le C_{n}^{k} \le T_{n}^{k}$ and $T_{n}^{k}>0$ for all $k$;
\item the time of arrival of a new task instance coincides with the absolute deadline of the previous task instance of the same task, i.e. $a_n^{k+1}=a_n^k+T_n^k$  for all $k$.
\end{enumerate}
\end{definition}

Figure \ref{fig:acyclictask} demonstrates an acyclic task. The horizontal line represents the progression of time. The upward arrows represent the times of arrival of new task instances, and the rectangles represent the computation of task instances. The computing times  and the relative deadlines are also marked. These plotting conventions will be followed by other figures in Section \ref{section:robustness}.
\begin{figure}[tp]
\begin{center}
\includegraphics[width=0.70 \textwidth]{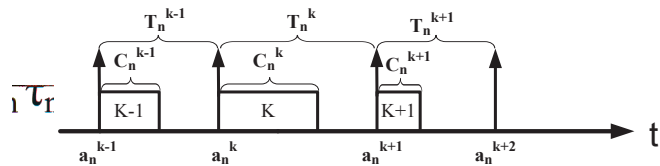}
\caption{Illustration of one acyclic task scheduled on a processor. Three task instances indexed by $k-1$, $k$, and $k+1$ are plotted.}  \label{fig:acyclictask}
\end{center}
\end{figure}

We use the acyclic task model because it  is universal: {(1)} any periodic task can be represented by an equivalent acyclic task. For example, a periodic task with computing time $2$ and period $5$ can be represented by an acyclic task with $C_{n}^{k}=2$ and $T_{n}^{k}=5$ for all $k$; {(2)} any set of non periodic tasks, i.e. tasks with irregular arriving instances, can be represented by an equivalent set of acyclic tasks \cite{Abdelzaher04}.

We want to model the scheduled behaviors of the real-time tasks at any time $t$. Some new notations that are only slightly different from the classical notations for acyclic tasks are necessary.
\begin{definition} \label{definition:effective}
At any time $t$, an instance of $\tau_n$ is \emph{effective} if and only if it has arrived before time $t$ but has not expired, i.e., $\tau_n^{k}$ is effective at time $t$ if and only if
\begin{equation}
a_n^{k}\leq t< a_n^{k}+T_n^{k}.
\end{equation}
\end{definition}

\begin{definition}
At any time $t$, $C_n(t)$ is defined as \emph{the computing time} of the effective instance of $\tau_n$ and $T_n(t)$ is defined as \emph{the relative deadline} of the effective instance of $\tau_n$, i.e.
\begin{equation} \label{equation:CT}
\begin{array}{ll}
C_n(t)=C_n^{k} \;\;\mbox{ and }\;\; T_n(t)=T_n^{k}\;\;  \mbox{ if }\;\;  a^{k}_{n}\le t <a_n^{k}+T_n^{k}.
\end{array}
\end{equation}
\end{definition}

\subsection{The Dynamic Timing Model} \label{sec:dynamictimingmodel}
In this section, we derive a mathematical model that describes the scheduled behaviors of a set of acyclic tasks within $[t_a, t_b]$ under any scheduling algorithm. We rely on the following assumption:
\begin{assumption} \label{assumption:perturbationknowledge}
At the starting  time $t_a$  we assume that  the values of $\{C_n(t)\}_{n=1}^{N}$ and $\{T_n(t)\}_{n=1}^{N}$ for $ t \in [t_a, t_b]$ are predictable.
\end{assumption}
Several key concepts will be defined including the state variables, the fixed priority window,  and the dynamic timing model.

\subsubsection{State Variables} \label{section:states}
The state variables are usually used to to derive differential or difference equations that describe dynamic systems behaviors \cite{Brogan}.   To describe the dynamic behaviors of scheduled tasks, we define two state variables and one auxiliary variable as follows.
\begin{definition} \label{definition:Q}
The {\em dynamic deadline} $Q(t)$  is defined as a vector $Q(t)=[q_1(t),$ $\dots, q_N(t)]$. Each  $q_n(t)$, for $n=1,2,...,N$, is the length of the time interval starting at the time instant $t$ and ending at the absolute deadline for the effective instance of $\tau_n$.
\end{definition}
In other words, suppose $\tau_n^k$ is an effective task instance, then $q_n(t)=a_n^k+T_n^k - t$ .

\begin{definition} \label{definition:S}
The {\em spare} $S(t)$ is  defined as a vector $S(t)=[s_1(t), ..., s_N(t)]$, where $s_n(t)$, for $n=1,2,...,N$, denotes the amount of  CPU time that is available to compute the effective instance of $\tau_n$ from its time of arrival  to time instant $t$.
\end{definition}

\begin{definition} \label{definition:R}
 The {\em residue} $R(t)$ is an auxiliary variable that is defined as a vector $R(t)=[r_1(t), ..., r_N(t)]$, where $r_n(t)$, for $n=1,2,...,N$, denotes the remaining computing time required {\em after} time $t$  to finish computing the effective instance of $\tau_n$.
\end{definition}

We use the following example to further explain the meaning of $Q$, $R$ and $S$. For  ease of demonstration, we consider three periodic tasks.
\begin{example} \label{example}
Consider tasks $\{\tau_1,\tau_2,\tau_3\}$ with $[C_1(t), C_2(t), C_3(t)]=[0.5, 1, 2]$ and $[T_1(t), T_2(t), T_3(t)]=[3, 4, 6]$ for $t \in [0, +\infty)$. The three periodic tasks are scheduled under a fixed priority preemptive scheduling algorithm such that the priority of $\tau_1$ is higher than $\tau_2$, and the priority of $\tau_2$ is higher than $\tau_3$.
\end{example}

%

\begin{figure}[tp]
\centering
\subfigure[Scheduled Behavior]{
\includegraphics[width=4in]{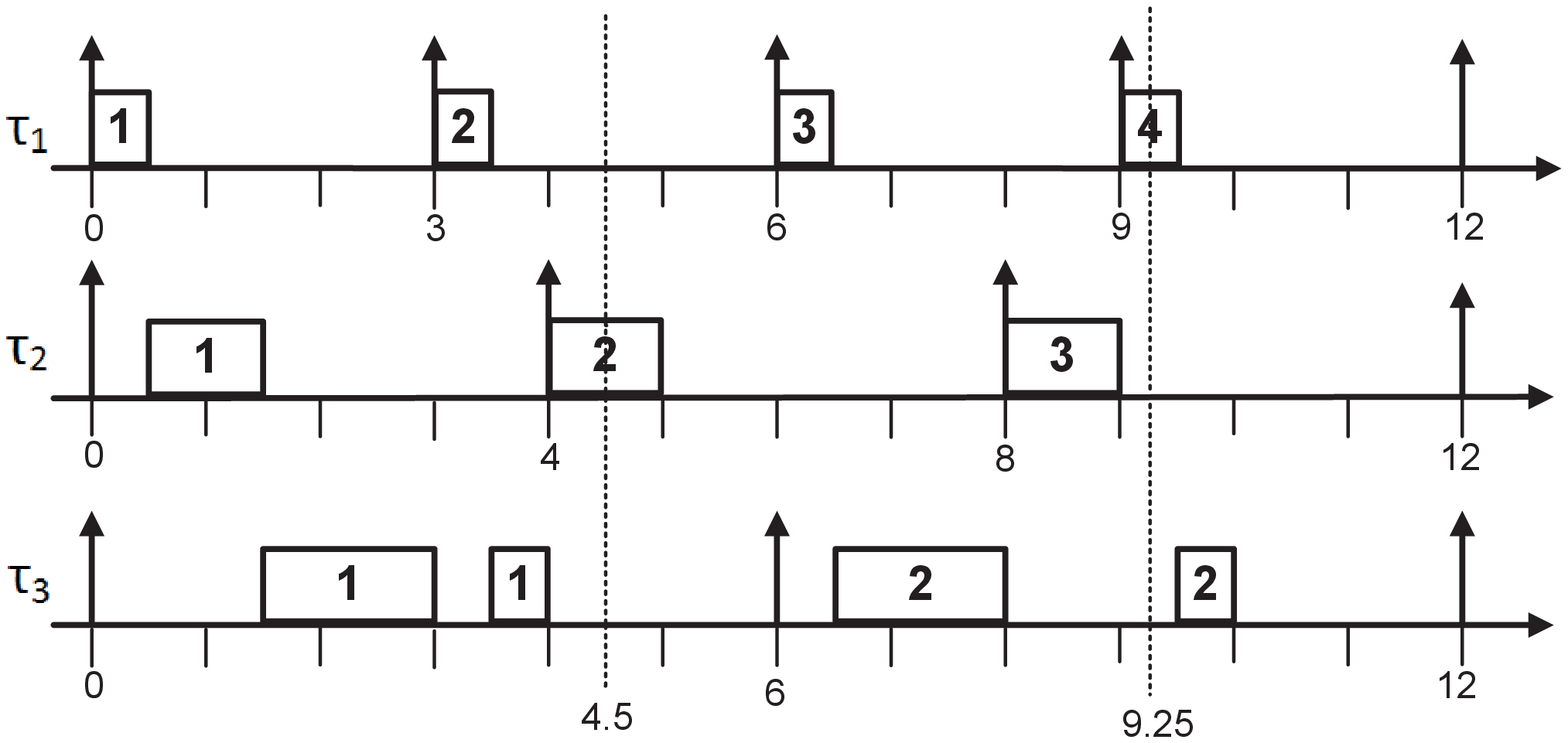}
\label{fig:statevariables}
}\\
\quad \subfigure[Fixed Priority Window]{
\includegraphics[width=3.9in]{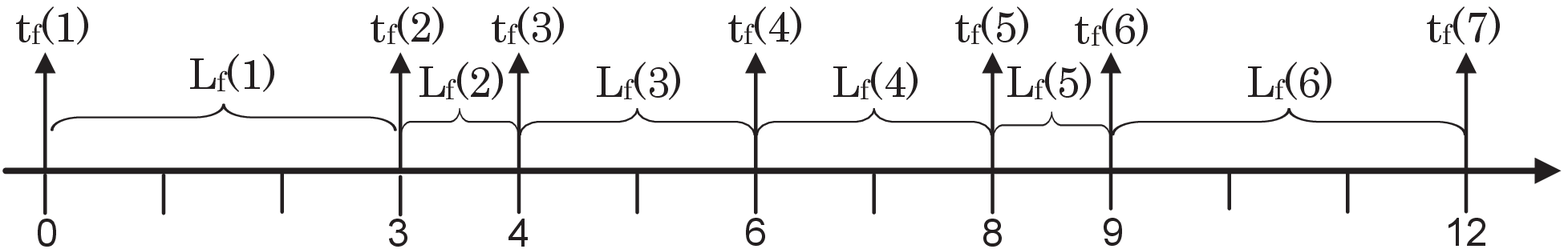}
\label{fig:fixedprioritywindow}
}
\caption{Three acyclic tasks scheduled on one processor}
\end{figure}

Figure \ref{fig:statevariables} demonstrates the computation of $\{\tau_1, \tau_2, \tau_3\}$ on one processor. We use the same plotting conventions as in Figure \ref{fig:acyclictask}, where the upper arrows indicate the times of arrival of the task instances.
It can be observed that the computation of lower priority tasks are interrupted by the computation of higher priority tasks. When $t=4.5$, $\tau_1^2$, $\tau_2^2$ and $\tau_3^1$ are the effective instances of the three tasks with time of arrival  $3$, $4$ and $0$ respectively.

We can observe that at  $t=4.5$,  $\tau_1^2$, $\tau_2^2$ and $\tau_3^1$ will expire at $6$, $8$ and $6$ respectively.
Thus, according to Definition \ref{definition:Q}, the relative deadlines are
\begin{equation}
[q_1(4.5), q_2(4.5), q_3(4.5)]=[6-4.5, 8-4.5, 6-4.5]=[1.5, 3.5, 1.5].
\end{equation}
After $t=4.5$ only $\tau_2^2$ has not finished computing. Therefore,  the remaining computing times after $t=4.5$ are $0$, $0.5$ and $0$. By Definition \ref{definition:R}, we have
\begin{equation}
 [r_1(4.5), r_2(4.5), r_3(4.5)]=[0, 0.5, 0].
\end{equation}
For $\tau_1^2$ with time of arrival at $3$, since no higher priority task is computed within $[3, 4.5]$, all the CPU time within $[3, 4.5]$ is available for $\tau_1^2$. For $\tau_2^2$ with time of arrival at $4$, since no higher priority task is computed within $[4, 4.5]$, all the CPU time within $[4, 4.5]$ is available for $\tau_2^2$. For $\tau_3^1$ with time of arrival $0$, since the CPU time within $[0, 1.5]$, $[3, 3.5]$ and $[4, 4.5]$ is allocated to the higher priority tasks, only the CPU time within $[1.5, 3]$ and $[3.5, 4]$ is available for $\tau_3^1$. Thus, according to Definition \ref{definition:S}, we have that
\begin{equation}
[s_1(4.5), s_2(4.5), s_3(4.5)]=[1.5, 0.5, 2].
\end{equation}
Similarly, at $t=9.25$, we can find
\begin{equation}
Q(9.25)=[2.75, 2.75, 2.75], \;\; R(9.25)=[0.25, 0, 0.5], \;\; S(9.25)=[0.25, 1, 1.5].
\end{equation}

It is worth mentioning that $s_n(t)$ is the amount of CPU time available to compute the effective instance of task $\tau_n$, but not necessarily the amount of CPU time actually taken by that instance. If  $s_n(t)\leq C_n(t)$, then the amount of CPU time
spent to compute the effective instance of task $\tau_n$ will be $s_n(t)$, which makes $r_n(t)=C_n(t)-s_n(t)$. On the other hand, if $s_n(t)>C_n(t)$, then the amount of CPU time spent to compute the effective instance of $\tau_n$ will only be $C_n(t)$, and the extra CPU time will be given to tasks with lower priority than $\tau_n$. In this case $r_n(t)$ will be zero since  no more computing time is needed. Therefore,
\begin{equation}
\label{equation:rsrelation}
r_n(t)={\rm max}\{0,C_n(t)-s_n(t)\}.
\end{equation}
This equation shows  that $R(t)$ solely depends on $S(t)$,  and explains why $R(t)$ is not a state variable. However,  $R(t)$ is more convenient to use for developing the dynamic timing model and the scheduled behavior in Section \ref{section:fixedprioritywindowbehavior} and Section \ref{section:scheduledbehavior}.

\subsubsection{Scheduling Algorithms}
We will now rigorously define a scheduling algorithm, which will be used by our mathematical models for the scheduled tasks later.
Let $\S=\{1,2,...,N\}$ be the set of  indices of tasks and let the function ${\rm Card}(\cdot)$ measure the number of elements in a set. Let $hp(n,t)$ denote the set of tasks with priorities higher than $\tau_n$ at time $t$. One way to formally define a scheduling algorithm is as follows.
\begin{definition}
A {\em scheduling algorithm} is a set-valued map between $\S \times R^+$ and the collection of all subsets of  $\S$. It is parametrized as $hp(n,t)$ where $n\in \S$ and $t \in \R^+$ so that $hp(n, t)\subset hp(m,t)$ if ${\rm Card}(hp(n,t))< {\rm Card}(hp(m,t))$.
\end{definition}

For example, assume all tasks are periodic and the RMS algorithm \cite{liulayland73} is used to assign fixed  priorities. Suppose that tasks are labeled according to the length of their periods i.e. tasks with longer periods have larger indices. Then  we  have:
\begin{equation}
\label{equation:fixedpriority}
hp(n,t)=\{1,2,...,n-1\}.
\end{equation}

Consider another example where a dynamic priority scheduling algorithm such as the EDF algorithm is used. Then, the values of $hp(n,t)$ depend on $Q(t)$. At any time $t$, the EDF  assigns higher priorities to the tasks whose effective instances have closer absolute deadlines. According to the definition of $Q(t)$, tasks whose effective instances having closer absolute deadlines also have smaller dynamic deadlines. Thus, for the EDF, the tasks with smaller values of $q_n(t)$ are assigned higher priorities. When two tasks have the same dynamic deadlines, we assume that a higher priority is assigned to the task with a smaller index. Hence, the set $hp(n,t)$ can be expressed as
\begin{equation}
hp(n,t)=\{i \vert  \mbox{either}\; q_i(t)< q_n(t), \; {\rm or} \; q_i(t)=q_n(t) \mbox{ and } i<n\}.
\end{equation}

\subsubsection{Fixed priority window}

Let us consider the time interval $[t_a, t_b]$ where the schedulibility of the tasks is concerned.  We further divide $[t_a, t_b]$ into
consecutive sub-intervals $[t_f(w), t_f(w+1))$, where $t_f(1)=t_a$ and $w=1,2,\dotsb$.
We require each sub-interval to be a fixed priority window as defined below:
\begin{definition} \label{definition:fixedprioritywindow}
A time interval $[t_f(w), t_f(w+1))$ is a {\em fixed priority window} if no instance of any task arrives within $(t_f(w), t_f(w+1))$.
\end{definition}
In other words, task instance can only arrive at either $t_f(w)$ or $t_f(w+1)$ but not in between.

To better understand this definition, we consider Figure 2(b) as an example: $[0,3)$ is a fixed priority window because no new instance of any task arrives within $(0,3)$; and $[0, 4)$ is not a fixed priority window because the task instance $\tau_1^2$ arrives at time $3 \in (0,4)$.

 The advantage of dividing $[t_a, t_b]$ into consecutive fixed priority windows  is that real-time tasks within each fixed priority window $[t_f(w), t_f(w+1))$ are relatively easier to be modeled.  These models can then be concatenated to derive more complex models for the scheduled behaviors on $[t_a, t_b]$.

Next, we study how to divide $[t_a, t_b]$ into consecutive fixed priority windows. We denote the  length of each window by $L_f(w)$, i.e
\begin{equation}
L_f(w)=t_f(w+1)-t_f(w),
\end{equation}
then each window $[t_f(w), t_f(w+1))$ can be rewritten as $[t_f(w), t_f(w)+L_f(w))$.
Hence, the partition of $[t_a, t_b]$ into fixed priority windows is determined by the window length $L_f(w)$ for $w=1,2,\dotsb$. To determine the value of each $L_f(w)$, we have the following claim
\begin{claim}
\label{claim:fixedprioritywindow}
For a set of acyclic tasks, at the beginning of any sub-interval, i.e. $t_f(w)$, if we choose $L_f(w)\leq \min\{q_1(t_f(w))$, $...,q_N(t_f(w))\}$, then $[t_f(w),t_f(w)+L_f(w))$ is a fixed priority window; otherwise, $[t_f(w),t_f(w)+L_f(w))$ is not a fixed priority window.
\end{claim}
\begin{proof}
At the beginning of any sub-interval, i.e. $t_f(w)$, consider the dynamic deadlines $Q(t_f(w))=[q_1(t_f(w)), ...$ $, q_N(t_f(w))]$, as defined in Definition \ref{definition:Q}. According to the definition of $Q(t_f(w))$, we know that the next task instance after $t_f(w)$ arrives at $t_f(w)+\min\{q_1(t_f(w)),...,q_N(t_f(w))\}$.

If we choose $L_{f}(w)=\min\{q_1(t_f(w)),...,q_N(t_f(w))\}$, then no new instance of any task arrives in between $(t_f(w), t_f(w)+L_f(w))$. Therefore, $[t_f(w), t_f(w)+L_f(w))$ is a fixed priority window.

On the other hand, if we choose $L_{f}(w)>\min\{q_1(t_f(w)),...,q_N(t_f(w))\}$, the next  task instance after $t_f(w)$ will arrive in between $(t_f(w), t_f(w)+L_f(w))$. Therefore, $[t_f(w), t_f(w)+L_f(w))$ is not a fixed priority window.
\end{proof}

The division of $[t_a, t_b]$ into consecutive fixed priority windows is carried out using the following procedure.  At the beginning of the first sub-interval, let $t_f(1)=t_a$, we choose the first window length $L(1)$ to make the sub-interval $[t_f(1), t_f(1)+L_f(1))$ a fixed priority window. Then by letting $t_f(2)=t_f(1)+L_f(1)$ and choosing a window length $L_f(2)$, the second sub-interval $[t_f(2), t_f(2)+L_f(2))$ can be made a fixed priority window. The  process is repeated untill one sub-interval reaches the ending time $t_b$.
According to Claim {\ref{claim:fixedprioritywindow}}, we know that the largest possible window length $L_f(w)$ can be expressed as
\begin{equation} \label{equation:windowlength}
 L_f(w)=\min\{q_1(t_f(w)),..., q_N(t_f(w)), t_b-t_f(w)\}
\end{equation}
where the extra term $t_b-t_f(w)$ guarantees that the division procedure stops at time $t_b$.  A larger window length is preferred since it reduces the complexity in modeling the behaviors of tasks. Figure \ref{fig:fixedprioritywindow} shows an example of dividing the time interval $[0, 12]$ into a series of consecutive fixed priority windows for Example \ref{example} discussed previously.

\subsubsection{Evolution of the state variables} \label{section:fixedprioritywindowbehavior}
With the state variables well defined in Section \ref{section:states}, we are now ready to define the dynamic timing model as follows:
\begin{definition}
The {\em dynamic timing model} is a set of equations that describes the evolution of the state variables over time $t$.
\end{definition}
For simplicity, we focus here on the evolution of the state variables within one fixed priority window $[t_f(w), t_f(w)+L_f(w))$. Later, the evolution of the state variables within any time interval $[t_a, t_b]$ can be obtained by concatenating the models within each fixed priority window that belongs to $[t_a, t_b]$. For notational simplicity, we will drop the index $w$. Moreover, we will use $t^{-}$ to denote the time point that is less than $t$ but is arbitrarily close to $t$. Thus, the fixed priority window $[t_f(w), t_f(w)+L_f(w))$ can now be equivalently written as $[t_f, \{t_f+L_f\}^{-}]$.

In the dynamic timing model, the evolution of the state variables $Q(t)$ and $S(t)$, from the end of the last fixed priority window $t_f^{-}$ to any time within the current fixed priority window $t\in[t_f, \{t_f+L_f\}^{-}]$, can be derived in two steps: from $t_f^{-}$ to $t_f$, and from $t_f$ to $t$.

\noindent {\bf From $t_f^{-}$ to $t_f$:}
First, we discuss the evolution for the state variables from $t_f^{-}$ to $t_f$. For task $\tau_n$, the values of the state variables at time $t_f$, denoted by $q_n(t_f)$ and $s_n(t_f)$, depend on whether an instance of $\tau_n$ arrives at $t_f$.

\noindent (1) if no instance of $\tau_n$ arrives at $t_f$ then the dynamic deadline for $\tau_n$ is unchanged and must be positive i.e. $q_n(t^-_f)> 0$, and all state variables hold their values from $t_f^-$ to $t_f$, i.e.,
\begin{equation} \label{equation: q1}
{\rm when}\;\; q_n(t^-_f)> 0: \quad q_n(t_f)=q_n(t_f^{-}) \;\;{\rm and}\;\; s_n(t_f)=s_n(t_f^{-}) \;\;\; .
\end{equation}
(2) if an instance of $\tau_n$ arrives at $t_f$ then the dynamic deadline for $\tau_n$ will be reset to $0$ at $t^-_f$ i.e. $q_n(t^-_f)=0$.  The dynamic deadline at $t_f$ will be the relative deadline for the new task instance i.e. $q_n(t_f)=T_n(t_f)$. The state spare $s_n(t_f)$ is reset to zero since no time is available between $t^-_f$ and $t_f$. Therefore, we have
\begin{equation} \label{equation: q2}
{\rm when}\;\; q_n(t^-_f)= 0: \quad q_n(t_f)=T_n(t_f) \;\;{\rm and}\;\;  s_n(t_f)=0.
\end{equation}
In summary, according to (\ref{equation: q1}) and (\ref{equation: q2}), the evolution for the state variables from $t_f^{-}$ to $t_f$ can be written in a compact form as follows
\begin{align}
\label{equation:evolution1}
q_n(t_f)&=q_n(t_f^{-})+T_n(t_f)(1-{\mathsf{sgn}}(q_n(t_f^{-}))) \cr
s_n(t_f)&=s_n(t_f^{-}){\mathsf{sgn}}(q_n(t_f^{-}))
\end{align}
where ${\mathsf{sgn}}$ denotes the signum function, i.e. ${\mathsf{sgn}}(x)=1$ when $x>0$, ${\mathsf{sgn}}(x)=0$ when $x=0$, and ${\mathsf{sgn}}(x)=-1$ when $x<0$.

\noindent{\bf From $t_f$ to $t$:}
Next, we discuss the evolution for the state variables from $t_f$ to $t \in [t_f, \{t_f+L_f\}^{-}]$.

\noindent (1) For the dynamic deadline $q_n(t)$,   we know that  the absolute deadline for the effective instance of $\tau_n$  is at $t+q_n(t)$. Since this absolute deadline is also at $t_f+q(t_f)$, we must have $q_n(t)+t=q_n(t_f)+t_f$. Therefore, the equation for $q_n(t)$ can be written as
\begin{equation}
\label{equation:dt}
q_n(t)=t_f+q_n(t_f)-t.
\end{equation}
(2) For the spare $s_n(t)$, we know that the computation of $\tau_n$ is preempted until the computation of all higher priority tasks are completed.  Then, the amount of time within $[t_f, t]$ that is available to compute $\tau_n$ is
\begin{equation}
\label{eq:timeavailable}
{\rm max}\{0, t-t_f-{\displaystyle \sum_{\substack{i\in hp(n,t_f)}}r_i(t_f)}\}.
\end{equation}
where $\sum_{\substack{i\in hp(n,t_f)}}r_i(t_f)$ denotes the time allocated to compute tasks with higher priorities than $\tau_n$. The function max guarantees that it will not give a negative result. Therefore, the amount of time that is available to compute the effective instance of  $\tau_n$ from its time of arrival to $t$ is
\begin{equation} \label{eq:dt2}
s_n(t)=s_n(t_f)+{\rm max}\{0, t-t_f-{\displaystyle \sum_{\substack{i\in hp(n,t_f)}}r_i(t_f)}\}.
\end{equation}
In summary, according to (\ref{equation:dt}) and (\ref{eq:dt2}), the evolution for the state variables from $t_f$ to $t \in [t_f, \{t_f+L_f\}^{-}]$ can be expressed as
\begin{equation} \label{equation:evolution2}
\begin{split}
q_n(t)&=t_f+q_n(t_f)-t \\
s_n(t)&=s_n(t_f)+{\rm max}\{0, t-t_f-{\displaystyle \sum_{\substack{i\in hp(n,t_f)}}r_i(t_f)}\}.
\end{split}
\end{equation}
where $r_i(t_f)={\rm max}\{0,C_i(t_f)-s_i(t_f)\}$ according to equation (\ref{equation:rsrelation}).

\begin{algorithm}[tp]
\caption{Model} \label{algorithm:Model}
/* when $t \in [t_f, \{t_f+L_f\}^{-}]$*/ \\
\KwData{$t_f$, $t$, $Q(t_f^{-})$, $S(t_f^{-})$, $\{C_n(t)\}_{n=1}^{N}$, $\{T_n(t)\}_{n=1}^{N}$}
\KwResult{$Q(t)$, $S(t)$}
\BlankLine
\lnl{}\For{each task $\tau_n \in \Gamma$}{
/*the value of $Q, S$ at $t_f$*/ \\
\lnl{} $q_n(t_f)=q_n(t_f^{-})+T_n(t_f)(1-{\mathsf{sgn}}(q_n(t_f^{-}))) $\;
\lnl{} $s_n(t_f)=s_n(t_f^{-}){\mathsf{sgn}}(q_n(t_f^{-}))$ \;
\lnl{} $r_n(t_f)=\max \{0, C_n(t_f)-s_n(t_f)\}$\;
/*the value of $Q, S$ at $t \in [t_f, \{t_f+L_f\}^{-}]$*/ \\
\lnl{}$q_n(t)=t_f+q_n(t_f)-t$ \;
\lnl{}$s_n(t)=s_n(t_f)+{\rm max}\{0,t-t_f-\sum_{\substack{i\in hp(n,t_f)}}r_i(t_f)\}$\;
}
\lnl{}\Return $Q(t),S(t)$\;
\end{algorithm}

The mathematical equations discussed in (\ref{equation:evolution1}) and (\ref{equation:evolution2}) constitute the dynamic timing model within one fixed priority window $[t_f, \{t_f+L_f\}^{-}]$, which can be implemented using Algorithm \ref{algorithm:Model}.  Given the initial values of the state variables at $t_{f}^-$,  i.e. $Q(t_f^{-})$ and $S(t_f^{-})$, and the task characteristics within the fixed priority window, i.e., $\{C_n(t)\}_{n=1}^{N}$ and $\{T_n(t)\}_{n=1}^{N}$ for $t \in [t_f,\{t_f+L_f\}^-]$, we can use Algorithm \ref{algorithm:Model} to obtain the evolution of the state variables from $t_f^{-}$ to any time $t \in [t_f, \{t_f+L_f\}^{-}]$. The dynamic timing model within any time interval $[t_a, t_b]$ can be achieved by iteratively applying Algorithm \ref{algorithm:Model} to all the fixed priority windows.

\subsubsection{ Scheduled Behaviors of Tasks} \label{section:scheduledbehavior}
We demonstrate how to use the dynamic timing model to describe the scheduled behaviors of the real-time tasks. Consider $\Gamma=\{\tau_1, \tau_2, \dotsb, \tau_N\}$,  we first describe scheduled behavior of  task $\tau_n$ from $\Gamma$.  Within each fixed priority window $[t_f,\{t_f+L_f\}^{-}]$, the  scheduled behavior of task $\tau_n$ may go through three modes that will be indicated by a function  $\Phi_n(t)$: \\

{\bf  The preempted mode:} the computation of the effective instance of $\tau_n$ is blocked by tasks with higher priorities. This behavior is indicated  by letting $\Phi_n(t)=0.5$. It starts from the beginning of the fixed priority window $t_f$ and lasts for the amount of time ${\rm min}\{\sum_{\substack{i\in hp(n,t_f)}}r_i(t_f), L_f\}$, which is the sum of the remaining computing time of all higher priority tasks;   \\

{\bf The execution mode:}  the effective instance of $\tau_n$ is being computed by the CPU. The scheduled behavior  is indicated by letting $\Phi_n(t)=1$. It starts right after the preempted mode and lasts until the computation of the effect instance of $\tau_n$ completes, which equals $t_f+{\rm min}\{\sum_{\substack{i\in hp(n,t_f)+\{n\}}}r_i(t_f), L_f\}$;\\

{\bf The free mode:} the computation of the effective instance of $\tau_n$ has completed and new instance has not arrived. The scheduled behavior is indicated by letting $\Phi_n(t)=0$. It starts right after the execution mode and lasts till the end of the fixed priority window.

In summary, the scheduled behavior of $\tau_n$ within one fixed priority window $[t_f,\{t_f+L_f\}^{-}]$ can be expressed as
\begin{equation}
\label{equation:Phi}
\begin{array}{c}
\Phi_n(t)=\\
\\
\left\{
\begin{array}{ll}
0.5, & t \in [\;\;t_f\;\;, \;\; t_f+{\rm min}\{{\displaystyle \sum_{\substack{i\in \;hp(n,t_f)}}r_i(t_f)}, L_f\} \;\;] \\
1,   & t \in (\;\;t_f+{\rm min}\{{\displaystyle \sum_{\substack{i\in\; hp(n,t_f)}}r_i(t_f)}, L_f\}\;\;,\;\; t_f+{\rm min}\{{\displaystyle \sum_{\substack{i\in \; hp(n,t_f)+\{n\}}}r_i(t_f)}, L_f\}\;\; ] \\
0, &  t \in (\;\;t_f+{\rm min}\{\displaystyle \sum_{\substack{i\in \;hp(n,t_f)+\{n\}}}r_i(t_f), L_f\}\;\;,\;\; \{t_f+L_f\}^{-}\;\;]
\end{array}
\right.
\end{array}
\end{equation}
where $r_i(t_f)={\rm max}\{0,C_i(t_f)-s_i(t_f)\}$.

As it shows, the scheduled behavior of $\tau_n$ within one fixed priority window $[t_f,\{t_f+L_f\}^{-}]$ can be described by the state variables within $[t_f,\{t_f+L_f\}^{-}]$. Applying the same methodology for all tasks in $\Gamma$, we can derive the scheduled behavior of the real-time system within $[t_f,\{t_f+L_f\}^{-}]$. As the fixed priority window propagates forward,  the state variables will evolve according to the dynamic timing model in Algorithm \ref{algorithm:Model}. With the state variables evolving from $t_a$ to $t_b$, we obtain the scheduled behavior of the real-time system over the time interval $[t_a, t_b]$.

\begin{figure}[tp]
\centering
\includegraphics[width=0.70 \textwidth]{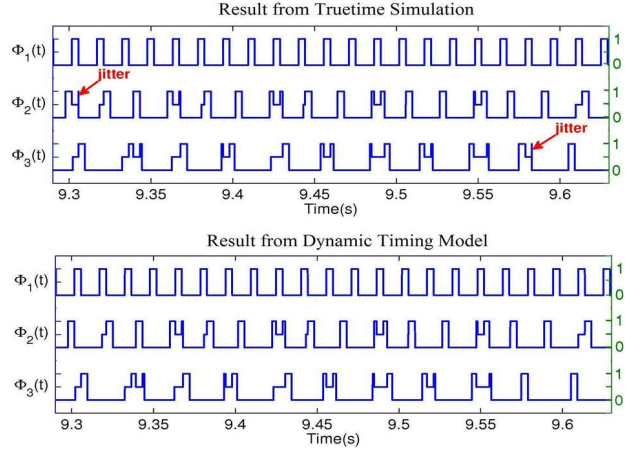}
\caption{The scheduled behaviors of $\Gamma$ within $[9.29, 9.63]$ seconds. The upper figure is produced by TrueTime, the lower figure is produced by the dynamic timing model. Jitters are marked by arrows.}  \label{fig:task execution}
\end{figure}

\subsubsection{Verification of the Dynamic Timing Model}
To verify the dynamic timing model, we compare the scheduled behavior of the real-time system derived from the dynamic timing model with the scheduled behavior of the same real-time system simulated using TrueTime \cite{truetime}. TrueTime is one of the most commonly used software tools that facilitates research on real-time systems. TrueTime and the dynamic timing model work in different ways. TrueTime simulates a computer with a real-time kernel and maintains data structures that are commonly found in the real-time kernel, such as ready queues, time queues, records for tasks, interrupt handlers, monitors, timers and so on \cite{truetime}. The dynamic timing model uses mathematical equations to analytically model the scheduling behavior, as shown in Algorithm \ref{algorithm:Model} and (\ref{equation:Phi}). For the same real-time system, ideally TrueTime and the dynamic timing model should provide the same result. However, we find incorrect jitters in the behavior generated by TrueTime 1.5 implemented in MATLAB. These jitters do not exist in the behavior generated by the dynamic timing model.

Suppose at time $0$, the state state variable  $Q(0^{-})=R(0^-)=0$. Consider a real-time system with three acyclic tasks running on it. The three acyclic tasks have the characteristics as $[C_1(t), C_2(t), C_3(t)]=[4, 4, 4]$ms and $[T_1(t), T_2(t), T_3(t)]=[15.4, 20.8, 30.3]$ms for $t \in [0, 10]$s. We are interested in the scheduled behavior of the real-time system within $[0, 10]$.  We run the simulation from $0$ to $10$s using TrueTime 1.5 implemented in MATLAB. Side by side, we evaluate the dynamic timing model and (\ref{equation:Phi}) using MATLAB from $0$ to $10$s. Figure \ref{fig:task execution} shows the comparative results of the scheduled behavior of the real-time tasks between the two different methods within $[9.29, 9.63]$.

By comparison, we see that the scheduled behaviors generated by TrueTime 1.5 and the dynamic timing model are identical  for most of the time. The identical part indicates that the dynamic timing model can be used to describe the scheduled behavior of the real-time system as precisely as TrueTime. However, the scheduled behaviors generated by TrueTime 1.5 and the dynamic timing model are not identical for $\Phi_2(t)$ when $t \in [9.3016, 9.3056]$s and for $\Phi_3(t)$ when $t \in [9.5788, 9.5828]$s. Further exploration shows that the differences are due to jitters caused by the numerical inaccuracy in TrueTime 1.5 implemented in MATLAB, as illustrated in the upper half of Fig.\ref{fig:task execution}. As a simulation tool, TrueTime 1.5 inevitably has truncation errors that accumulate with numerical integration. Since the dynamic timing model presented in this paper is based on mathematical equations, the system behavior at time $t$ can be determined by evaluating functions without using numerical integration.  Hence the chances for jitters are significantly reduced.  No jitters are observed from the lower half of Fig. \ref{fig:task execution}. This indicates that the dynamic timing model may be used side by side with TrueTime to resolve jitters.

\subsection{Dynamic Schedulability Test} \label{sec:schedulabilitytest}
In Section  \ref{sec:dynamictimingmodel},  we have established a dynamic timing model that can analytically describe the evolution of the state variables from $t_a$ to $t_b$. In this section, we study how to utilize the dynamic timing model to perform the dynamic schedulability test over $[t_a, t_b]$. The success of this test requires the knowledge of the task sets within $[t_a, t_b]$, as stated in Assumption \ref{assumption:perturbationknowledge}.

For the set of  real-time tasks $\Gamma=\{\tau_1, \tau_2, \dotsb, \tau_N\}$, the dynamic schedulability test over $[t_a, t_b]$ can be decomposed to check whether each task $\tau_n$ of $\Gamma$ is able to meet its deadlines within each fixed priority window that belongs to $[t_a, t_b]$. This is due to the following facts: (1) $\Gamma$ is schedulable within $[t_a, t_b]$ if and only if $\Gamma$ is schedulable within each fixed priority window $[t_f(w), \{t_f(w)+L_f(w)\}^{-}]$,  for $w=1,2,\dotsb$; (2) $\Gamma$ is schedulable within any fixed priority window $[t_f(w), \{t_f(w)+L_f(w)\}^{-}]$ if and only if each individual task $\tau_n \in \Gamma$ is schedulable within $[t_f(w), \{t_f(w)+L_f(w)\}^{-}]$. The following theorem states the necessary and sufficient conditions for the schedulability of $\tau_n$ within any fixed priority window $[t_f(w), \{t_f(w)+L_f(w)\}^{-}]$.
\begin{theorem}
\label{theorem:dynamicschedulability}
A task $\tau_n$ is schedulable within $[t_f(w), \{t_f(w)+L_f(w)\}^{-}]$ if and only if it satisfies ONE of the following two conditions:
\begin{enumerate}
\item $q_n(\{t_f(w)+L_f(w)\}^{-})=0$ and $C_n(\{t_f(w)+L_f(w)\}^{-})\le s_n(\{t_f(w)+L_f(w)\}^{-})$;
\item $q_n(\{t_f(w)+L_f(w)\}^{-})>0$.
\end{enumerate}
\end{theorem}

\begin{proof}
If an instance of $\tau_n$ expires at $t_f(w)+L_f(w)$, i.e. $q_n(\{t_f(w)+L_f(w)\}^{-})=0$, then the schedulability of $\tau_n$ within $[t_f(w), \{t_f(w)+L_f(w)\}^{-}]$ is satisfied if and only if  the computation of this instance has completed, i.e.
\begin{equation*}
r_n(\{t_f(w)+L_f(w)\}^{-})=0.
\end{equation*}
According to (\ref{equation:rsrelation}), the above equation can be rewritten as
\begin{equation*}
{\rm max}\{0,C_n(\{t_f(w)+L_f(w)\}^{-})-s_n(\{t_f(w)+L_f(w)\}^{-})\}=0,
\end{equation*}
which implies that
\begin{equation}
C_n(\{t_f(w)+L_f(w)\}^{-})\le s_n(\{t_f(w)+L_f(w)\}^{-}).
\end{equation}

If no instance of $\tau_n$ expires at $t_f(w)+L_f(w)$, i.e. $q_n(\{t_f(w)+L_f(w)\}^{-})>0$, then the schedulability of $\tau_n$ within $[t_f(w), \{t_f(w)+L_f(w)\}^{-}]$ is automatically guaranteed.
\end{proof}

According to Assumption \ref{assumption:perturbationknowledge}, we can predict the actual task characteristics $\{C_n(t)\}_{n=1}^{N}$ and $\{T_n(t)\}_{n=1}^{N}$ within $ [t_a, t_b]$.
Given the actual task characteristics $\{C_n(t)\}_{n=1}^{N}$ and $\{T_n(t)\}_{n=1}^{N}$ for $t \in [t_a, t_b]$, we can perform the dynamic schedulability test over the time interval $[t_a, t_b]$ using Algorithm \ref{algorithm:robustness}. Algorithm \ref{algorithm:robustness} iteratively checks the schedulability of $\Gamma$ within each fixed priority window in the following ways: (1) first, at the beginning of any sub-interval, it calculates the length of the current fixed priority window $L_f$ according to equations (\ref{equation:windowlength}), as shown in Lines $10$ of Algorithm \ref{algorithm:robustness}. (2) then, it utilizes the dynamic timing model in Algorithm \ref{algorithm:Model} to obtain the values of the state variables at the end of the current fixed priority window, as indicated by Line $11$; (3) finally, it evaluates the schedulability of $\tau_n$, where $n=1,\dotsb,N$, within $[t_f, \{t_f+L_f\}^-]$ according to Theorem \ref{theorem:dynamicschedulability}, as shown in  Lines $12-20$ of Algorithm \ref{algorithm:robustness}. To make the fixed priority window propagates seamlessly within $[t_a, t_b]$, it assigns the starting time of the next fixed priority window to be the ending time of the current fixed priority window, as indicated by Line $20$.

The variable ${\rm ds}_n(w)$ indicates the dynamic schedulability test result of $\tau_n$ within $[t_f(w), \{t_f(w)+L_f(w)\}^{-}]$: when $\tau_n$ is schedulable within $[t_f(w), \{t_f(w)+L_f(w)\}^{-}]$, ${\rm ds}_n(w)=1$; otherwise, ${\rm ds}_n(w)=0$. The set ${\rm DS}_n=[ds_n(1), ds_n(2), \dotsb]$ contains the dynamic schedulability test results of $\tau_n$ within all fixed priority windows that belong to $[t_a, t_b]$. The task $\tau_n$ is schedulable within $[t_a, t_b]$ if and only if  ${\rm min}\{\rm DS_n\}=1$.  The task set $\Gamma$ is schedulable within $[t_a,t_b]$ if and only if all individual tasks are dynamically schedulable within $[t_a, t_b]$, i.e. ${\rm min}_{1\le n \le N}\{{\rm min}\{\rm DS_n\}\}=1$.

\begin{algorithm}[tp]
\caption{Dynamic Schedulability Test}\label{algorithm:robustness}
/*Schedulability of $\Gamma$ within $[t_a,t_b]$ */ \\
\KwData{$t_a$, $t_b$, $Q(t_a^{-})$, $S(t_a^{-})$, $\{C_n(t)\}_{n=1}^{N}$, $\{T_n(t)\}_{n=1}^{N}$}
\KwResult{$\{\rm DS_n\}_{n=1}^{N}$}
\BlankLine
\lnl{}$t_f=t_a$\;
\lnl{}\For{each $\tau_n \in \Gamma$}{
\lnl{}${\rm DS_n}=[\;]$\;
}
/*check each fixed priority window*/ \\
\lnl{}\While{$t_f<t_b$}{
/* The length of the current fixed priority window $L_f$ */ \\
\lnl{}\For{each $\tau_n \in \Gamma$}{
\lnl{}\If{$q_n(t_f^{-})==0$}{
\lnl{}$q_n(t_f)=T_n(t_f)$\;}
\lnl{}\Else{
\lnl{}$q_n(t_f)=q_n(t_f^-)$\; }}
\lnl{}$L_f=\min\{q_1(t_f), ..., q_N(t_f), t_b-t_f\}$\;
/* State Variables at the end of the current fixed priority window */\\
\lnl{}$\;\;\;\;\;\;\;$ $[Q(\{t_f+L_f\}^-),S(\{t_f+L_f\}^-)]= $\\
Model($t_f, \{t_f+L_f\}^-, Q(t_f^-), S(t_f^-), \{C_n(t)\}_{n=1}^{N}, \{T_n(t)\}_{n=1}^{N}$)\;
/* Schedulability within the current fixed priority window */ \\
\lnl{}\For{each $\tau_n \in \Gamma$}{
\lnl{}\If{$q_n(\{t_f+L_f\}^{-})==0$}{
\lnl{}\If{$C_n(\{t_f+L_f\}^{-})<s_n(\{t_f+L_f\}^{-})$}{
\lnl{}${\rm ds_n}=1$;}
\lnl{}\Else{
\lnl{}${\rm ds_n}=0$\;}
}
\lnl{}\Else{
\lnl{}${\rm ds_n}=1$\;}
\lnl{} ${\rm DS_n}=[{\rm DS_n, ds_n}]$ \;}
\lnl{}$t_f=t_f+L_f$ \;
}
\lnl{}\Return $\{\rm DS_n\}_{n=1}^{N}$\;
\end{algorithm}

\subsection{A Measure of Robustness} \label{section:schedulingrobustness}
We let  $\{C_n^{\rm nom}(t)\}_{n=1}^{N}$ and $\{T_n^{\rm nom}(t)\}_{n=1}^{N}$ denote the nominal task characteristics known at the design phase, and let $\{C_n(t)\}_{n=1}^{N}$ and $\{T_n(t)\}_{n=1}^{N}$ denote the actual task characteristics under online perturbations.
We assume that there is no perturbation on the relative deadlines, i.e. $T_n(t)=T_n^{\rm nom}(t)$ for $n=1,2,..., N$. This assumption is reasonable in control and robotics applications,  where $T_n(t)$ represent sampling times that are often fixed.
At time $t$, we define the (instantaneous) perturbations on computing times as follows:
\begin{definition}
The {\em perturbations on computing times} are defined as a vector $\mathcal{E}(t)=[\epsilon_1(t),..., \epsilon_N(t)]$, where $\epsilon_n(t)=C_n(t)-C^{\rm nom}_n(t)$ for $n=1,2, ..., N$.
\end{definition}
 The value of $\epsilon_n(t)$ can be either positive or negative. If $C_n(t)>C^{\rm nom}_n(t)$, then $\epsilon_n(t)$ is positive.
Note that in future works, $T_n(t)$ may be viewed as a control variable that can be adjusted to tolerate the perturbations in similar ways as the general elastic scheduling algorithms \cite{Buttazzo02,HuLemmon06}.

Next, we consider the accumulated effect caused by  the perturbations $\mathcal{E}(t)$ over time. These effects will be captured by defining
perturbations on  the state variables.  We let
 $\{Q_n^{\rm nom}(t)\}_{n=1}^{N}$ and $\{S_n^{\rm nom}(t)\}_{n=1}^{N}$ denote the state variables in the nominal case, and let $\{Q_n(t)\}_{n=1}^{N}$ and $\{S_n(t)\}_{n=1}^{N}$ denote the state variables under accumulated perturbations.
 Since $T_n(t)=T_n^{\rm nom}(t)$ for $n=1,2,..., N$, we know that the absolute deadline and the time of arrival of each task instance in the nominal case is the same as these in the actual case. Thus, according to Definition \ref{definition:Q}, we know that the dynamic deadline of each task instance in the nominal case is the same as that in the actual case, i.e.
\begin{equation}\label{equation:perturedQ}
q_n(t)=q_n^{\rm nom}(t)
\end{equation}
which, together with (\ref{equation:windowlength}),  implies that
\begin{equation}\label{equation:perturedQextend}
t_f(w)=t_f^{\rm nom}(w) \quad L_f(w)=L_f^{\rm nom}(w).
\end{equation}
On the other hand, since $C_n(t)\neq C_n^{\rm nom}(t)$, we know that the spare of each task instance in the nominal case is different from that in the actual case, i.e.
\begin{equation}\label{equation:perturedS}
s_n(t) \neq s_n^{\rm nom}(t).
\end{equation}
Equations (\ref{equation:perturedQ}) and (\ref{equation:perturedS}) indicate that there are perturbations on the state variable $S$, but not on the state variable $Q$. We define the perturbations on the state variable $S$ as follows:
\begin{definition}
The {\em perturbations on the state variable spare} is defined as a vector $\mathcal{H}(t)=[\eta_1(t),..., \eta_N(t)]$, where $\eta_n(t)$ denotes the strength of the perturbation on $s_n(t)$, i.e.
\begin{equation}
\label{equation:eta}
\eta_n(t)=-(s_n(t)-s^{\rm nom}_n(t))
\end{equation}
where we use a negative sign because a positive perturbation imposed on the computing time of a task instance will reduce the value of the spare.
\end{definition}

According to the above analysis, we know that at any time $t$, the total perturbations imposed on the real-time tasks consist of two portions: $\mathcal{E}(t)$, the perturbations on the computing time, and $\mathcal{H}(t)$, the perturbations on the state variable spare, which reflects the accumulated effect of $\mathcal{E}(t)$ before time $t$.
The total perturbations imposed on the real-time system at time $t$ are the summation $\mathcal{E}(t)+\mathcal{H}(t)$.

 In particular, the total perturbations imposed on one task $\tau_n$ at time $t$ can be expressed as $\epsilon_n(t)+\eta_n(t)$.  We are interested in finding the maximum total perturbations $\epsilon_n(t)+\eta_n(t)$ that can be tolerated by a single task $\tau_n$ without sacrificing the schedulability of $\tau_n$. According to (\ref{equation:perturedQ}), (\ref{equation:perturedQextend}) and Theorem \ref{theorem:dynamicschedulability}, we can easily prove the following claims.
\begin{claim} \label{claim:perturbation}
$\tau_n$ is schedulable within $[t_f(w), \{t_f(w)+L_f(w)\}^{-}]$ under  perturbations $\epsilon_n(t)+\eta_n(t)$  if and only if  ONE of the following two conditions are satisfied:
\begin{enumerate}
\item $q_n(\{t_f(w)+L_f(w)\}^{-})=0$ and $\epsilon_{n}(\{t_f(w)+L_f(w)\}^{-})+\eta_{n}(\{t_f(w)+L_f(w)\}^{-})\le s^{\rm nom}_n(\{t_f^{\rm nom}(w)+L_f^{\rm nom}(w)\}^{-})-C^{\rm nom}_n(\{t_f^{\rm nom}(w)+L_f^{\rm nom}(w)\}^{-})$;
\item $q_n(\{t_f(w)+L_f(w)\}^{-})>0$.
\end{enumerate}
\end{claim}

We introduce a measure of robustness $B_R$ that quantifies the tolerance of a real-time scheduling algorithm to uncertain perturbations to the computing times of tasks within $[t_a, t_b]$. A real-time scheduling algorithm with a larger value for $B_R$ is more robust than a real-time scheduling algorithm with smaller values for $B_R$.
\begin{definition} \label{df:BR}
We define a {\em measure of robustness} $B_R(w)$ over the fixed priority window $[t_f(w), \{t_f(w)+L_f(w)\}^{-}]$ where $w=1,2,...$ as the least upper bound on the tolerable perturbations for all task instances expiring at $t_f(w)+L_f(w)$, i.e.
\begin{equation} \label{equation:BRW}
\begin{array}{c}
B_R(w)={\rm min}_{n \in \{i\vert q_i(\{t_f(w)+L_f(w)\}^{-})=0\}}\\
(s^{\rm nom}_n(\{t_f^{\rm nom}(w)+L_f^{\rm nom}(w)\}^{-})-C^{\rm nom}_n(\{t_f^{\rm nom}(w)+L_f^{\rm nom}(w)\}^{-}))
\end{array}
\end{equation}
We define the measure of robustness $B_R$ over time interval $[t_a, t_b]$ as the minimum value of $B_R(w)$  i.e.
\begin{equation}
B_R= \min_w B_R(w).
\end{equation}
\end{definition}

\begin{claim}
Within $[t_a, t_b]$, the  nominal design of an acyclic task set under a real-time scheduling algorithm is schedulable under any perturbation of a strength less than $B_R$.
\end{claim}
\begin{proof}
Suppose an arbitrary task $\tau_n$ suffers the perturbation $\epsilon_{n}(\{t_f(w)+L_f(w)\}^{-})$ $+\eta_{n}(\{t_f(w)+L_f(w)\}^{-})$ at the end of a fixed priority window $[t_f(w), \{t_f(w)+L_f(w)\}^-]$. If $q_n(\{t_f(w)+L_f(w)\}^{-})>0$, the second condition in Claim \ref{claim:perturbation} is satisfied and $\tau_n$ is schedulable under the perturbation; if $q_n(\{t_f(w)+L_f(w)\}^{-})=0$, we have that $\epsilon_{n}(\{t_f(w)+L_f(w)\}^{-})+\eta_{n}(\{t_f(w)+L_f(w)\}^{-})\leq B_R \leq B_R(w) \leq s^{\rm nom}_n(\{t_f^{\rm nom}(w)+L_f^{\rm nom}(w)\}^{-})-C^{\rm nom}_n(\{t_f^{\rm nom}(w)+L_f^{\rm nom}(w)\}^{-})$. Thus, the first condition in Claim \ref{claim:perturbation} is satisfied and $\tau_n$ is schedulable to the perturbation. Since the above proof holds for any task within any fixed priority window that belongs to $[t_a, t_b]$, the nominal design  is schedulable under any perturbation of a strength less than $B_R$.
\end{proof}

At any time $t_a$, if we input the nominal task characteristics $\{T_n^{\rm nom}(t)\}_{n=1}^{N}$ and $\{C_n^{\rm nom}(t)\}_{n=1}^{N}$ to Algorithm \ref{algorithm:Model},  we can obtain the evolution of the nominal state variables $\{Q_n^{\rm nom}(t)\}_{n=1}^{N}$ and $\{S_n^{\rm nom}(t)\}_{n=1}^{N}$ from $t_a$ to $t_b$ by iteratively applying the dynamic timing model in Algorithm \ref{algorithm:Model}. Moreover, the right hand side of (\ref{equation:BRW}) is computed at $t_a$ by using the nominal state variables. Therefore, the measure of robustness of the real-time system $B_R$ can be predicted at $t_a$ without relying on Assumption \ref{assumption:perturbationknowledge}.

\section{Robustness in Battery Management}\label{batteryman}

Robustness of  a battery management algorithm can be measured by its tolerance to potentially harmful discharges and variations in battery parameters.  The tolerance decreases when the  SoC decreases as the battery is being drained.
Battery management algorithms can be developed to manage multiple batteries at the same time,  so that a battery near the point of depletion can be replaced by a freshly charged battery. We will show that  the
SoC of a battery can be estimated
at any point of time during system operation  using the combination of a dynamic battery model and the dynamic timing model developed in the previous section. We further present an algorithm to predict whether the battery is capable of maintaining a steady output voltage when it is supporting a time-varying load. The methodology used to detect impending battery failure can be used in any battery management system to increase robustness.

\subsection{Background}
\subsubsection{Dynamic Battery Model}
Battery modeling is a challenging task due to complex electro-chemical processes occurring within a battery \cite{rvw_model,rvw_vlsi}.  Battery models can be represented in various forms. Chen and Mora \cite{rincon} provide models that are verified by experimental data and are more suitable to be combined with our dynamic timing model.

\begin{figure}[tp]
      \centering
      \includegraphics[width=0.9\textwidth]{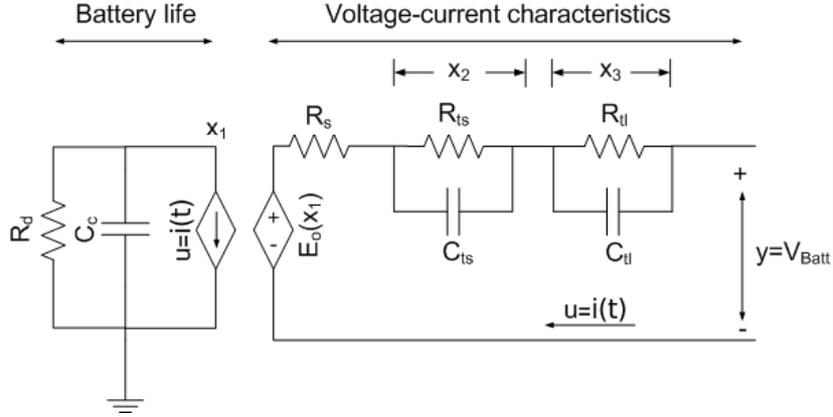}
            \caption{Chen and Mora's battery model}
      \label{battmdl}
\end{figure}

Chen and Mora's model as shown in Figure \ref{battmdl} is an equivalent circuit representation of a Lithium-ion (Li-ion) battery. The model has two coupled circuits. The circuit on the left models the
SoC $x_1$
and the  circuit on the right models the variation of the battery output voltage $y$ as a function of the charge/discharge current $i(t)$.
It must be noted that all the circuit components $C_{ts}, C_{tl}, R_{s}, R_{ts},  R_{tl}, E_o, C_c$ are  nonlinear functions of $x_1$ as follows:
\begin{eqnarray}
 C_{ts}&=&-k_4e^{-k_1x_1}+k_3\label{e13} \\
 C_{tl}&=&-k_6e^{-k_2x_1}+k_5\label{e15}\\
 R_s&=&k_7e^{-k_8x_1}+k_9\label{e11}\\
 R_{ts}&=&k_{10}e^{-k_{11}x_1}+k_{12}\label{e12}\\
 R_{tl}&=&k_{13}e^{-k_{14}x_1}+k_{15}\label{e14}\\
 E_{o}&=&-k_{16}e^{-k_{17}x_1}+k_{18}+k_{19}x_1\\\nonumber
      &-&k_{20}{x_1}^2+k_{21}{x_1}^3\label{e16}\\
 C_c&=&3600Cf_1f_2. \label{e17}
\end{eqnarray}
where $k_i>0$ for $i=1,2,...,21$.
In eqn. \eqref{e17} $f_1,f_2\in[0,1]$ are factors taking into account the effects of temperature and charge-discharge cycles respectively. By default, $f_1=f_2=1$, but their values will decrease after each charge-discharge cycle.
The various resistances, capacitances, and constants ($k_1,\cdots,k_{21}$) shown here are independent of $i(t)$. Hence it enables one to experimentally determine these parameters at different stages during the life of a battery \cite{rincon,AbuSharkh,Bernhardt,Coleman}. The experimental data justifies that the model can be applied to applications with acceptable accuracy.

Knauff et.al. \cite{Knauff} provide a state space realization for the above battery model. We have introduced minor modifications to aid our analysis.
\begin{eqnarray}\label{e8}
 \dot{x}_1&=&-\frac{1}{C_c} i\label{e8.2}\\
 \dot{x}_2&=&-\frac{x_2}{R_{ts}C_{ts}}+\frac{i}{C_{ts}}\label{e8.3}\\
 \dot{x}_3&=&-\frac{x_3}{R_{tl}C_{tl}}+\frac{i}{C_{tl}}\label{e8.4}\\
 y&=&E_o-x_2-x_3- i R_s, \label{e8.5}
\end{eqnarray}
where  $y$ represents the voltage output from the battery,  $x_2$ represents the voltage drop across $R_{ts}||C_{ts}$, and  $x_3$ represents the voltage drop across $R_{tl}||C_{tl}$.
\begin{figure}[tp]
\centering
\subfigure[V versus $x_1$]{
      \includegraphics[width=0.44\textwidth, height=5.00cm]{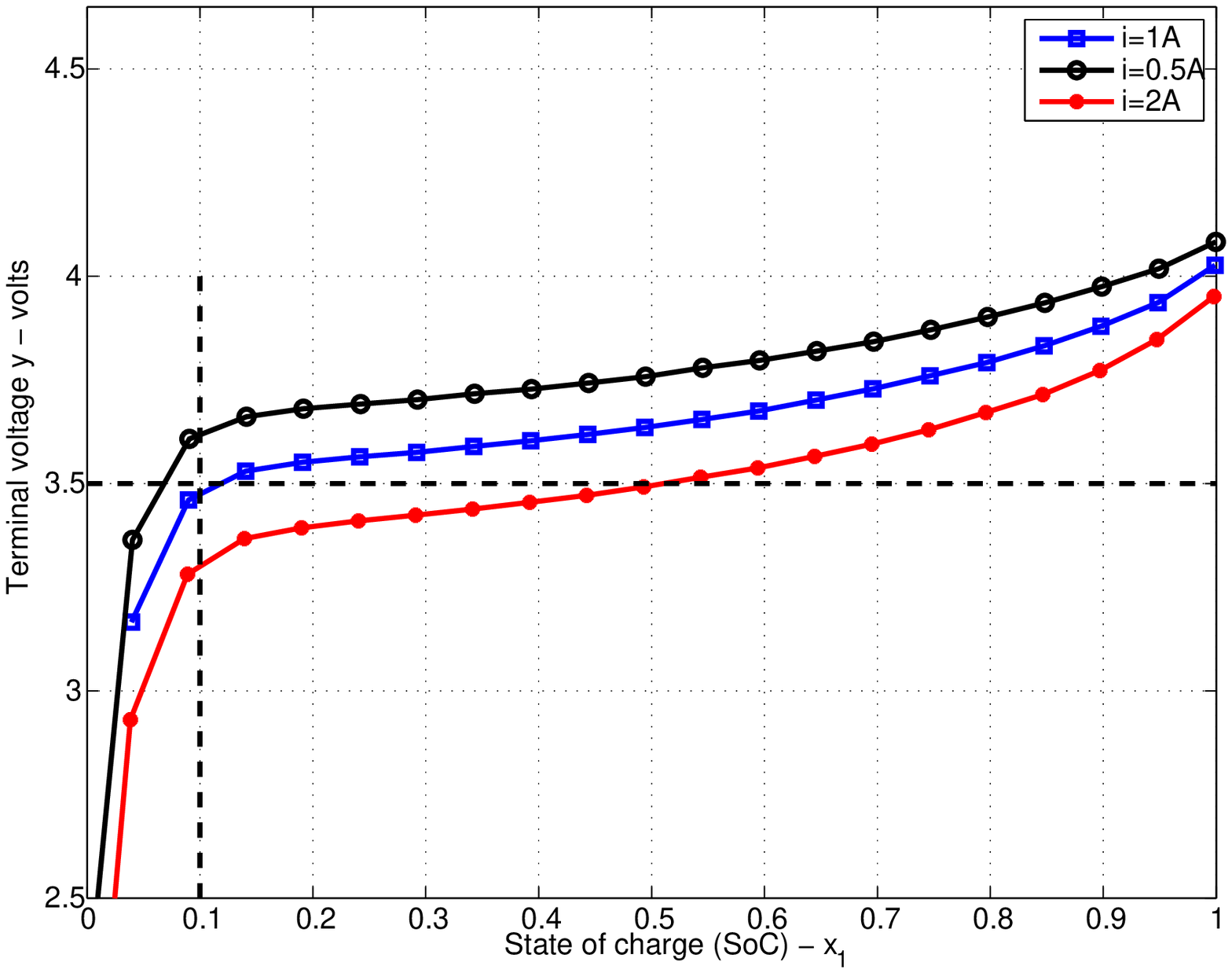}
      \label{vxt}
}
\subfigure[V versus t]{
      \includegraphics[width=0.49\textwidth]{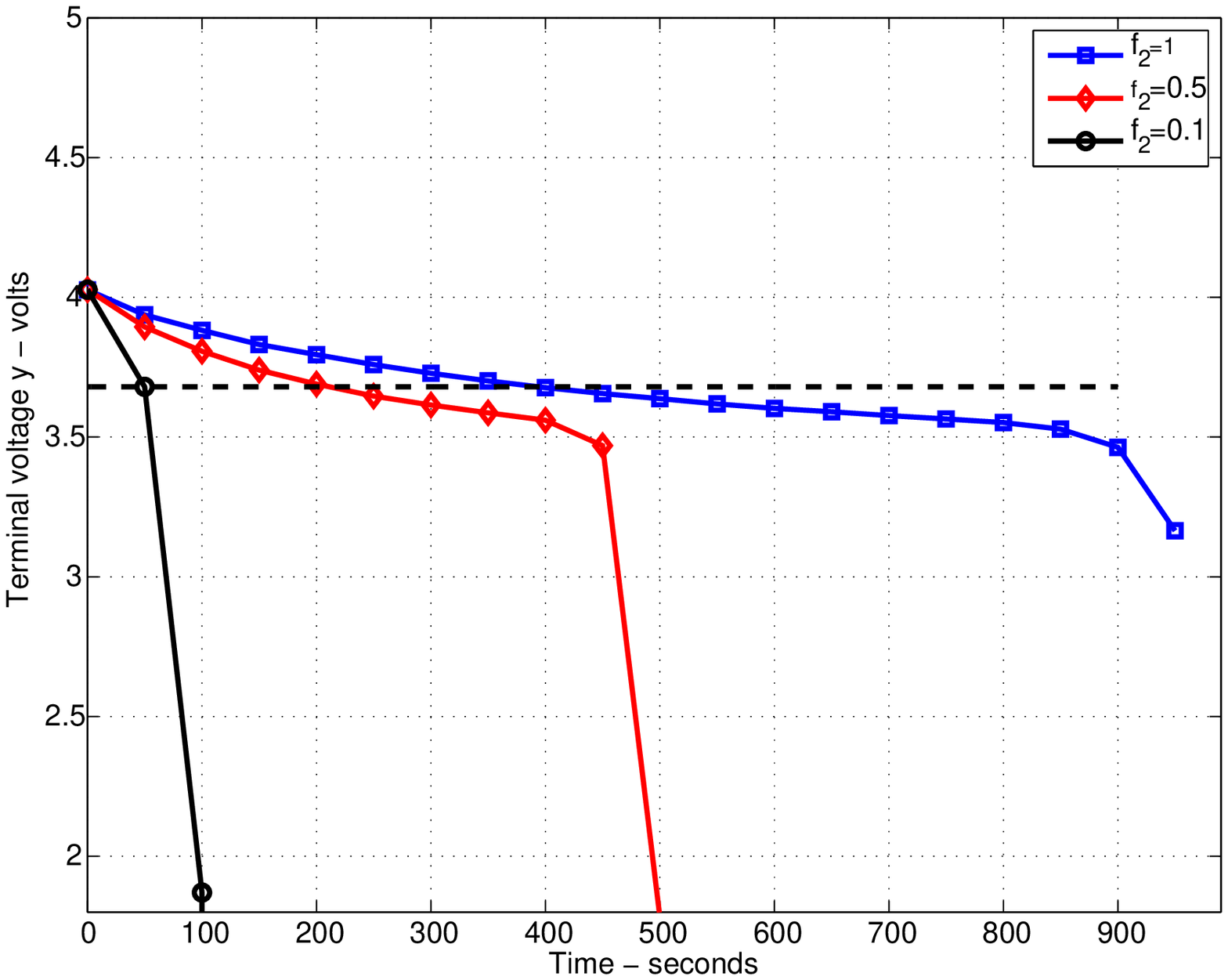}
      \label{vbt}
}
\caption[Caption for subfigures]{The characteristics of battery voltage variations. (a) The changes in battery voltage as a function of SoC for different constant current loads. (b) The variation of battery voltage with respect to time `$t$' for different values of $f_2$ at the same constant current load.}
\end{figure}

\subsubsection{Voltage Thresholding and Capacity Thresholding} Figures \ref{vxt} and \ref{vbt} show typical battery characteristics. One important problem is how to detect battery failure based on these characteristic curves.
The horizontal and vertical dashed lines represent static thresholds on the terminal voltage and the SoC respectively. The Voltage Thresholding (VT) method detects battery failure when the output voltage of the battery drops below a threshold represented by the horizontal line. The Capacity Thresholding (CT) method detects battery failure when the SoC of the battery drops below a threshold represented by the vertical line.

Each curve in Fig. \ref{vxt} shows the relation between the SoC and terminal voltage for a specific constant value of the discharge current. For a load current of $0.5A$ or $1A$ the voltage threshold of $3.7V$ (shown by the horizontal dashed line) detects battery failure when the battery voltage starts declining rapidly. However, for a load current of $2A$, VT detects failure with SoC still at $50\%$. Assuming that the voltage has not fallen below the operational requirements of the system, this would result in switching a battery out of service unnecessarily. The vertical dashed line in fig. \ref{vxt} shows an SoC threshold of $0.1$. For loads of $1A$ and  $2A$,  CT detects failure correctly. But for a lighter load of $0.5A$, CT detects failure even though the terminal voltage is higher than the previously set threshold. Thus the battery is switched out earlier than necessary in this case.

Figure \ref{vbt} shows the variation of battery voltage with respect to time $t$ for different values of $f_2$ at the same constant current load. The horizontal dashed line represents a voltage threshold of $3.5V$. When $f_2=0.1$, VT based on this threshold detects failure right before the terminal voltage starts declining rapidly. However if $f_2=0.5$ or $1$,  VT switches out the battery early since the figure shows that the terminal voltage does not start dropping rapidly for a long time after failure is detected.

VT and CT are  generally used to detect battery failure \cite{BattHandBook,kim_rtas,philipsbook}. From figures \ref{vxt} and \ref{vbt} it is obvious that changes in the load current $i$ and $f_2$ can cause static thresholds to be overly conservative. This can cause batteries to be switched out of the system when there may be a significant amount of usable capacity available. We call this phenomena the {\it false alarm}.  False alarms will reduce the operational life of battery supported systems and increase maintenance cost.

We will design a new algorithm, called the Adaptive Thresholding (AT),  which is able to determine an adaptive threshold that adjusts automatically to the changes in the battery parameters. This further leads us to the notion of robustness of battery switching algorithms.

\subsection{Battery Stability}
We observe that the battery system represented by eqns. \eqref{e8.2}-\eqref{e8.5} looses stability (in the sense of control theory) when the battery terminal voltage  drops suddenly.
Consider the state $x_1$ as a parameter. Temporarily disregarding the input $i$, the system in eqns. \eqref{e8.2}-\eqref{e8.5} can be rewritten using standard state space notation \cite{dis} as the following non-autonomous system,
\begin{equation}
 \begin{bmatrix}{\dot{x}_2}\cr
                {\dot{x}_3}\cr
             \end{bmatrix}=A({x_1})\begin{bmatrix}
                            {x_2}\cr
                            {x_3}\cr
             \end{bmatrix} \mbox{  where } A({x_1})=\begin{bmatrix}{\frac{-1}{C_{ts}R_{ts}}}&{0}\cr
        {0}&{\frac{-1}{C_{tl}R_{tl}}}
        \end{bmatrix}\label{exfunc1}.
\end{equation}
The above representation simplifies the nonlinear model of a battery to a linear time-varying model.

Consider $C_{ts}$ and $C_{tl}$ for our battery model where $k_1,\cdots,k_6$ satisfy the condition $
 0<k_1<k_2<k_3<k_4<k_5<k_6$.
 Regarding  eqn. \eqref{exfunc1},  our first stability result is based on the following candidate Lyapunov function and its time derivative:
\begin{eqnarray}
 V_1&=&\frac{1}{2}(x_2^2+x_3^2)\label{lyapfunc1}\\
 \dot{V}_1&=&-\bigg(\frac{x_2^2}{R_{ts}C_{ts}}+\frac{x_3^2}{R_{tl}C_{tl}}\bigg).\label{lyapderiv1}
\end{eqnarray}
\begin{lemma}\label{thm-lyapthm1}
Consider $C_{ts},C_{tl},R_{ts},R_{tl}$, $V_1$, and $\dot{V}_1$ in equations \eqref{e13}-\eqref{e14}, \eqref{lyapfunc1}, and \eqref{lyapderiv1}  respectively.
Assuming that $\frac{1}{k_1}\ln\big(\frac{k_3}{k_4}\big)>\frac{1}{k_2}\ln\big(\frac{k_5}{k_6}\big)$, for the
SoC
$x_1 \in [0,1]$ and discharge current $i(t)>0$, there exist small positive numbers $\{(\delta_1,\delta_2)|0<\delta_1<\delta_2\}$ such that
	$\dot{V}_1 > 0$ for  $x_1 \in (0,\delta_1)$ and $\dot{V}_1 \leq 0$ for $x_1 \in (\delta_2, 1]$.
\end{lemma}

\begin{proof}\label{proof-lyapthm1}
We observe that $V_1>0$, for all $x_2,x_3 \neq 0$. Since
$R_{ts},R_{tl}$ have the form $ae^{-bx_1}+c$, where $a,b,c>0$, then
	$R_{ts},R_{tl} >0$ for all  $x_1$. Consider the case when $C_{ts}<0$. Solving eqn. \eqref{e13} for $x_1$ gives,
$
	x_1<-\frac{1}{k_1}\ln\bigg(\frac{k_3}{k_4}\bigg).
$ Similarly, considering $C_{tl}<0$ and solving eqn. \eqref{e15} for $x_1$ gives
\begin{equation}
	x_1<-\frac{1}{k_2}\ln\bigg(\frac{k_5}{k_6}\bigg).   \label{ctlcond1}
\end{equation} Let us define $\delta_1$ and $\delta_2$ as follows,
\begin{eqnarray}
 \delta_1&=&-\frac{1}{k_1}\ln\bigg(\frac{k_3}{k_4}\bigg)\label{d1},\\
 \delta_2&=&-\frac{1}{k_2}\ln\bigg(\frac{k_5}{k_6}\bigg).\label{d2}
\end{eqnarray}Since $k_3<k_4$ and $k_5<k_6$, we have $\delta_1,\delta_2>0$. Based on our assumptions we further have,
$0<\delta_1<\delta_2$. Therefore, if $x_1<\delta_1$ then $C_{ts}, C_{tl}<0$, which makes $\dot{V}_1$ positive. Similarly if $x_1>\delta_2$ then $C_{ts}, C_{tl}>0$ and $\dot{V}_1$ is negative. We have proved the existence of $\delta_1$ and $\delta_2$.
\end{proof}
From the above proof, it is observed that the battery is unstable (in the Lyapunov sense \cite{Khalil}) when $x_1\in (0,\delta_1)$. When $x_1 \in (\delta_2, 1]$ the battery is stable. $\delta_1$ thus provides the worst case limit for
the SoC
of a battery. If the
SoC
falls below $\delta_1$, one must switch a battery out of service, otherwise the output voltage will soon drop below any specified bound.
Note that the representation in eqn. \eqref{exfunc1} simply aids in establishing the stability limits and is not used to explicitly replicate the dynamics. Hence it does not introduce any error. These limits are applicable even to the system in eqns. \eqref{e8}-\eqref{e8.5}.

The following claim can be made based on the previous lemma.
\begin{claim}\label{corol-lyapthm1}
If $x_1<\delta_2$, where $\delta_2$ is obtained from lemma \ref{corol-lyapthm1}, then the Li-ion battery system represented by equation \eqref{exfunc1} is not asymptotically stable.
\end{claim}
\begin{proof}\label{proof-corol-lyapthm1}
 From eqns. \eqref{exfunc1}-\eqref{d2} it is obvious that if $x_1<\delta_2$, the two eigenvalues of $A({x_1})$ do not have  negative real parts. Hence the system is not asymptotically stable.
\end{proof}
This claim indicates that switching out a battery when $x_1<\delta_2$ is safer than switching out the battery later when $x_1<\delta_1$. Therefore, $\delta_2$ can now be viewed as a threshold for the
SoC of a battery
to indicate when a battery needs to be switched out. Note that $\delta_2$ does not depend on the discharge current $i(t)$.

Next, we develop an adaptive threshold that depends on $i(t)$.
We consider the nonlinear battery model represented by eqns. \eqref{e8.2}-\eqref{e8.4} with the input current $i(t)$. Let us consider the following candidate Lyapunov function and its time derivative.
\begin{eqnarray}
 V_2&=&\frac{1}{2}(x_1^2+x_2^2+x_3^2)\label{lyapfunc2}\\
 \dot{V}_2&=&i\bigg(\frac{x_2}{C_{ts}}+\frac{x_3}{C_{tl}}-\frac{x_1}{C_c}\bigg)-\bigg(\frac{x_2^2}{R_{ts}C_{ts}}+\frac{x_3^2}{R_{tl}C_{tl}}\bigg). \label{lyapderiv2}
\end{eqnarray}

\begin{lemma}\label{thm-lyapthm2}
Consider $C_{ts}, C_{tl}, R_{ts}, R_{tl}$, $V_2$, and $\dot{V}_2$ defined in eqns. \eqref{e13}-\eqref{e14},  \eqref{lyapfunc2}, and \eqref{lyapderiv2}. Consider $\delta_2$ obtained from Lemma \ref{thm-lyapthm1}. For the
SoC
$x_1 \in [0,1]$  and $R_{ts}, R_{tl}$, $C_{ts}, C_{tl}, x_2, x_3 >0$,  there exist a small positive lower bound $\epsilon(x_2,x_3)$ for the discharge current $i(t)$ and a threshold $\beta(x_2,x_3,i)$ for $x_1$ such that $\delta_2<\beta<1$ and the following two statements hold: \textbf{(1)} $\dot{V}_2 > 0$ if $x_1<\beta$ and $i>\epsilon$; \textbf{(2)} $\dot{V}_2 \leq 0$, if $x_1 \geq \beta$ and $i>\epsilon$.
\end{lemma}
\begin{proof}\label{proof-lyapthm2}
Considering $\dot{V}_2>0$ we have,
\begin{equation}
 i\bigg(\frac{x_2}{C_{ts}}+\frac{x_3}{C_{tl}}-\frac{x_1}{C_c}\bigg)-\bigg(\frac{x_2^2}{R_{ts}C_{ts}}+\frac{x_3^2}{R_{tl}C_{tl}}\bigg)>0\label{gr1}.
\end{equation}
Solving eqn. \eqref{gr1} for $x_1$ gives,
\begin{equation}
 x_1<C_c\Bigg(\frac{x_2}{C_{ts}}+\frac{x_3}{C_{tl}}-\frac{1}{i}\bigg(\frac{x_2^2}{R_{ts}C_{ts}}+\frac{x_3^2}{R_{tl}C_{tl}}\bigg)\Bigg)\label{gr2}.
\end{equation}
Let us define the quantity on the right-hand side of eqn. \eqref{gr2} as $\beta$,
\begin{equation}
	\beta \triangleq C_c\Bigg(\frac{x_2}{C_{ts}}+\frac{x_3}{C_{tl}}-\frac{1}{i}\bigg(\frac{x_2^2}{R_{ts}C_{ts}}+\frac{x_3^2}{R_{tl}C_{tl}}\bigg)\Bigg)\label{betadef}.
\end{equation}
From eqns. \eqref{gr2} and \eqref{betadef} we have $\dot{V_2}>0$ when $x_1<\beta$. Similarly, we can see that $\dot{V_2}\leq 0$ when $x_1\geq\beta$.

From eqn. \eqref{betadef} it is obvious that for very small positive values of the discharge current $i$, the value of $\beta$ will turn out negative. Solving eqn. \eqref{betadef} for the current $i$ when $\beta=0$ provides the lower bound $\epsilon$ for the discharge current.
\begin{equation}
 \epsilon=\bigg(\frac{x_2^2}{R_{ts}C_{ts}}+\frac{x_3^2}{R_{tl}C_{tl}}\bigg)\bigg/\bigg(\frac{x_2}{C_{ts}}+\frac{x_3}{C_{tl}}\bigg)\label{epsi}
\end{equation}
As per claim \ref{corol-lyapthm1}, stability of the battery system requires $x_1\geq\delta_2$. Hence we proceed to prove $\beta>\delta_2$ by contradiction.
Let us temporarily assume that $\beta\leq \delta_2$. Hence from eqn. \eqref{gr2} we have $x_1\leq\delta_2$. However, from eqns. \eqref{ctlcond1} and \eqref{d2} we have that $C_{tl} \leq 0$ if $x_1\leq \delta_2$. Thus assuming $\beta\leq \delta_2$ contradicts the condition $C_{tl}>0$. Hence by contradiction we have $ \beta > \delta_2$.
Thus proving the existence of $\epsilon(x_2,x_3)$  and $\beta(x_2,x_3,i)$.
\end{proof}

The above result provides an adaptive threshold $\beta$ for $x_1$.  Adaptive control theory \cite{krstic} serves as an inspiration for this design.
The threshold $\beta$ dynamically adjusts itself to account for the number of charge-discharge cycles and varying current. Since $\beta>\delta_2$, $\beta$ provides a more conservative threshold than $\delta_2$ for switching a battery out of service. From eqn. \eqref{betadef} we see that the states $x_2$ and $x_3$ are required to calculate $\beta$, while $\beta$ gives the threshold for $x_1$. Hence all the three states need to be estimated. We discretize the model given by eqns. \eqref{e8.2}-\eqref{e8.5} and run a particle filter to estimate the battery states. Satisfactory results from the particle filter have been observed, which are not presented in this paper since they are less relevant.
Particle filtering is one of many approaches to state estimation. We use particle filtering because of the presence of nonlinearities in the battery system. Although computationally complex, the emerging new generation multi-core embedded systems may offer the required computational capability. Other methods like extended Kalman filtering (EKF) \cite{Barbarisi06} which are computationally simpler can be used, although it may result in early/late switching out of a battery due to errors in the estimates.

\subsection{Robust Battery Switching\label{bm3}}

Claim \ref{corol-lyapthm1} provides the threshold $\delta_2$ for $x_1$ below which at least one of the eigenvalues of $A(x_1)$ has a positive real part. We have shown that when $x_1<\delta_2$, the battery will become unstable, indicating that the condition of the battery has degraded. We can use this threshold for measuring robustness of battery switching algorithms.
Variations in the battery discharge, the SoC, and the parameters can be viewed as perturbations to battery management algorithms.
\begin{definition} \label{brobust}
  A battery management  algorithm is {\em robust} if it guarantees that at the switching time instant when the battery is replaced,  the SoC of the battery
is above the threshold $\delta_2$ e.g. $x_1\geq\delta_2$.
\end{definition}

We develop a robust and adaptive switching algorithm, called the Adaptive Thresholding (AT),  to switch out batteries close to the end of their lives. In Algorithm \ref{algorithm:switch} we use the following quantities: $h$ is the sampling interval in seconds, $k$ is the time step at which the discharge current $i(t)$ and the battery output voltage $V$ are measured, $\tau_s$ is the battery switching time instant and $S=1$ indicates switching is necessary.

Our battery switching algorithm based on Lemma \ref{thm-lyapthm2} provides a threshold $\beta$. This threshold $\beta$ adjusts itself to perturbations in the SoC and the battery parameters so that $\beta<\delta_2$ is always satisfied. Hence our algorithm is robust by Definition \ref{brobust}.

\begin{algorithm}[htbp]
\caption{Determine Battery Switching Time Instant $\tau_s$}\label{algorithm:switch}
\KwData{$y(k), i(k), \epsilon$}
\KwResult{$S=[0,1], \tau_s$}
\BlankLine
\lnl{}$[\hat{x}_1,\hat{x}_2,\hat{x}_3]=ParticleFilter(y(k),i(k))$\;
\lnl{} Compute $\beta$ and $\epsilon$ using equations \eqref{betadef} and \eqref{epsi}\;
\lnl{}\If{$i(k)>\epsilon$}{
\lnl{}\If{$\hat{x}_1<\beta$}{
\lnl{}$S=1,\tau_s=hk$;}
\lnl{}\Else{
\lnl{}$S=0,\tau_s=-1$;}}
\lnl{}\Return $S,\tau_s$\;
\end{algorithm}

\section{Application} \label{sec:application}
To demonstrate the relevance of the robustness analysis for CPSb, we study a simplified scenario as shown in Figure \ref{fig:CPS}. Processor 1 issues control commands to the motors on the bases of multiple inverted pendulums.
Processor 2 runs the dynamic schedulability test and evaluates the particle filter that estimates the
SoC
of the battery based on measurements taken for the terminal voltage and the discharge current. We assume that Processor 2 implements the dynamic schedulability test described in section \ref{sec:schedulabilitytest} and the battery management algorithm described in section \ref{bm3}. When the
SoC
of a battery is below a specific threshold, the working battery will be disconnected and the other fully charged battery is switched in. We simulate this scenario since it simplifies real systems where computing of real-time control tasks are typically separated from battery management circuits.  Performing the schedulability test on a second processor can reduce the overhead on the first processor, where the real-time tasks are scheduled. The separation can be implemented by a dual processor system with the ability of programming each processor independently.

\begin{figure}[tp]
\centering
\includegraphics[width=0.4\textwidth]{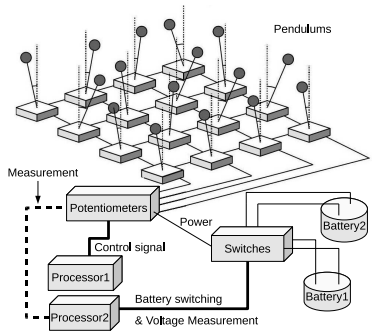}
\caption{A two battery-powered bi-processor system controlling multiple pendulums with different physical parameters.} \label{fig:CPS}
\end{figure}

The separation of the control and battery management on different processors does not conflict with the spirit of co-design. In fact, the control and scheduling on Processor 1 determines the battery discharge current that will affect the battery management
algorithm on Processor 2. Through simulations based on this system, we demonstrate robustness of the system subject to both timing perturbations and discharge perturbations.

\subsection{Real-time Tasks and Currents}
Suppose three pendulums are controlled by  control signals $u_1$, $u_2$ and $u_3$. These  control signals are computed using methods in \cite{ZWS_RTSS08}.
The three controllers implemented on Processor 1 can be viewed as three independent real-time tasks $\Gamma=\{\tau_1, \tau_2, \tau_3\}$ that need to be scheduled.

At the design phase, we assume that $\{\tau_n\}_{n=1}^{3}$ are periodic tasks with the nominal computing times  $[C^{\rm nom}_{1}(t), C^{\rm nom}_{2}(t), C^{\rm nom}_{3}(t)]=[4, 4, 4]$ms and they are scheduled under the RMS algorithm. By solving a minimization problem as introduced in \cite{ZWS_RTSS08}, we can determine the task periods to be $[T^{\rm nom}_{1}(t), T^{\rm nom}_{2}(t), $ $T^{\rm nom}_{3}(t)]=[15.4, 20.8, 30.3]$ms. In this scenario, the task periods are fixed once chosen, i.e. $T_{n}(t)=T^{\rm nom}_n(t)$ for $n=1,2,3$.  The control signals are kept constant during one task period and only updated at the end of each period. However, during runtime, $\{C_n(t)\}_{n=1}^{3}$ may deviate from $\{C^{\rm nom}_n(t)\}_{n=1}^{3}$ due to online perturbations. Moreover, if a task cannot finish the computation by its deadline, the control output will not update at the end of this period.

Assume that the online perturbations on the computing time $\{C^{\rm nom}_n(t)\}_{n=1}^{3}$ are generated from a stochastic processes $\mathcal{E}(t)$ with their value at each point in time being random variables that are uniformly distributed within $[-1.5, 4]$ms, $[-1, 4]$ms and $[-1, 2]$ms. Suppose the sample value of $\mathcal{E}(t)$ within $[10,13]$s are known at time $t=10$s. Then, we have the actual task characteristics $[T_1(t), T_2(t), T_3(t)]=[T_1^{\rm nom}(t), T_2^{\rm nom}(t), T_3^{\rm nom}(t)]$ and $[C_1(t), C_2(t), C_3(t)]=[C_1^{\rm nom}(t)+\epsilon_1(t), C_2^{\rm nom}(t)+\epsilon_2(t), C_3^{\rm nom}(t)+\epsilon_3(t)]$ for $t\in[10, 13]$. To check the schedulability under the perturbations, the scheduled behavior of the real-time system is shown in Fig \ref{fig:RMSbehavior},  and the result of  Algorithm \ref{algorithm:robustness} is shown in Fig.\ref{fig:RMSschedulability}. In Fig \ref{fig:RMSbehavior}, we observe that the value of $\Phi_3(t)$ does not fall back to zero before its deadline at $t=11.8475$s, which implies that the computation of $\tau_3$ fails to finish by its deadline. As we can see from the result of the dynamic schedulability test, ${\rm DS}_3(t)=0$ when $t\in [11.817, 11.8475]$s, which indicates that $\tau_3$ is  not schedulable within $[11.817, 11.8475]$s.

\begin{figure}[tp]
\centering
\subfigure[Scheduled Behavior of $\Gamma$]
{\includegraphics[width=0.40 \textwidth]{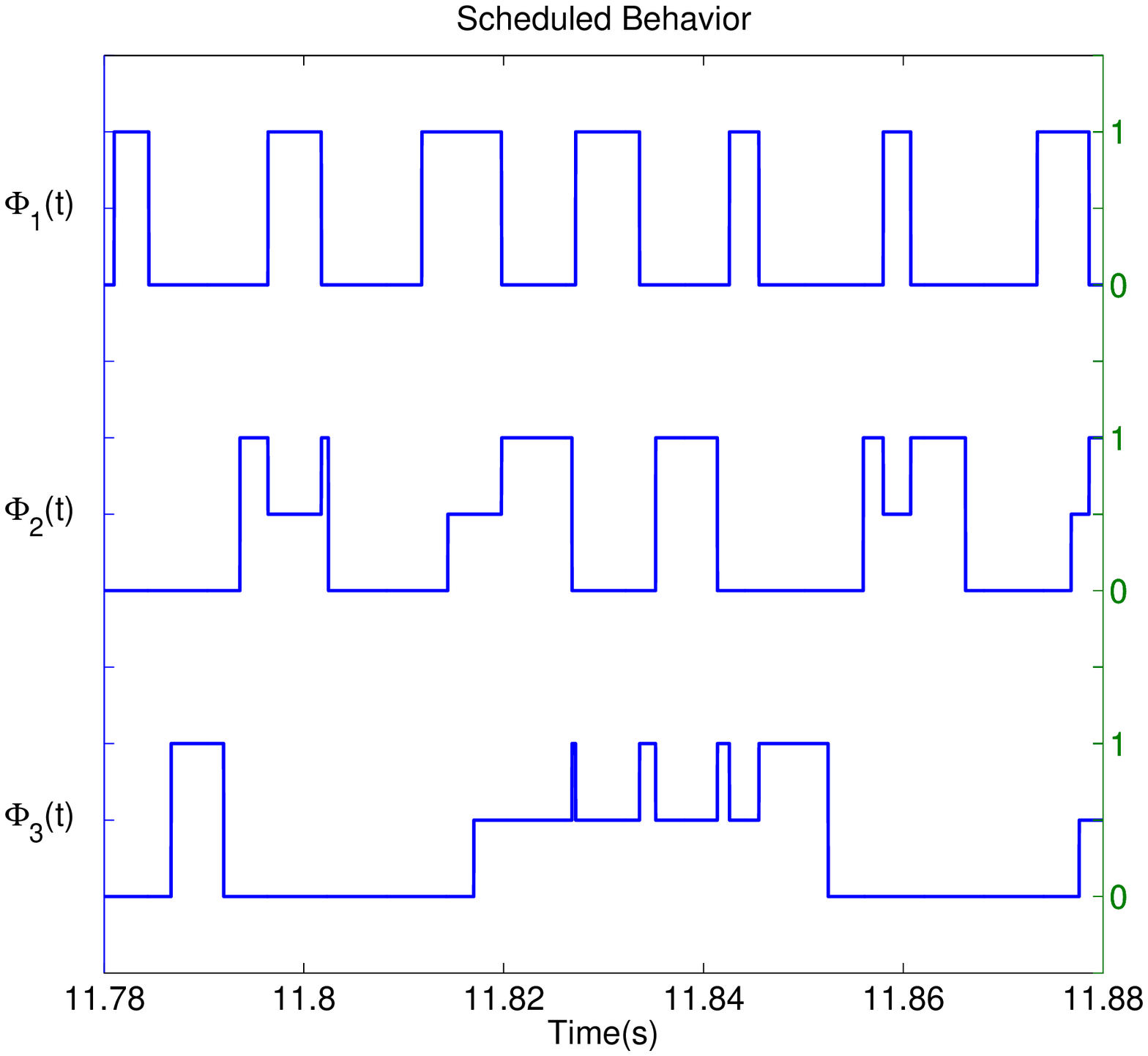}
  \label{fig:RMSbehavior}
 }
\subfigure[Dynamic Schedulability Test for $\Gamma$]
{\includegraphics[width=0.40 \textwidth]{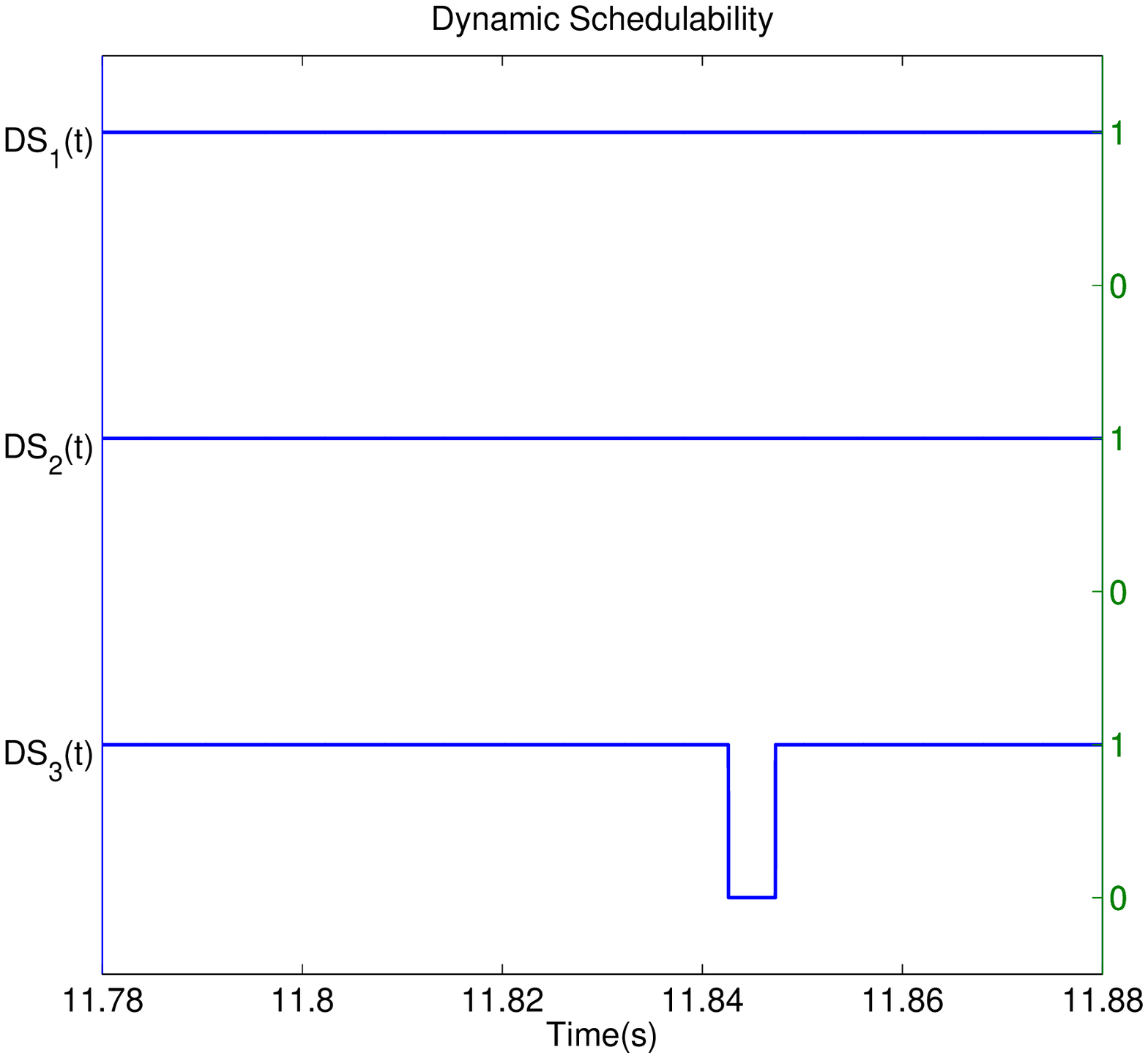}
  \label{fig:RMSschedulability}
 }
 \caption{The pendulum system under the RMS algorithm  subject to perturbations $\mathcal{E}(t)$}
\end{figure}

We assume that the pendulums are powered by permanent magnet DC shunt motors. The motors provide torque directly proportional to the current supplied \cite{machines}. The total load current drawn (ideally) from the battery can be written as:
$i_{tot}=P(|u_1|+|u_2|+|u_3|)+i_{p_1}+i_{p_2}$.
We explain each term and how they are determined:
\begin{enumerate}
 \item   $P$ is the constant of proportionality relating the torque to the current drawn. For simplicity we assume that the constant is the same for the three motors. We also choose $P=0.1$ for purposes of simulation. In reality this constant will change based on motor parameters and needs to be determined experimentally.

 \item  We assume that the first processor consumes an average of 400mA when it is computing and 200mA when it is idle. Hence the current absorbed by the first processor is $i_{p_1}=(300+100\Phi_{\rm cpu})$mA, as shown in Fig. \ref{fig:ip1}. It is easy to verify that the result in Fig. \ref{fig:ip1} is consistent with the result of Fig. \ref{fig:RMSbehavior} in that $\Phi_{\rm cpu}= sgn(\Phi_1+\Phi_2+\Phi_3)$.

 \item We assume that Processor 2 consumes $i_{p_2}=300mA$ constantly.
  \end{enumerate}

Using the dynamic timing model and the controller models, we can predict the total load current  supplied by the battery within $[10, 13]$s  at time 10s, as shown in Fig. \ref{totload}. In real life the current waveform may have small transient effects that are ignored here. We want to emphasize that all our methods developed in this paper and in \cite{ZWS_RTSS08} are analytical, hence the waveforms can be obtained {\it analytically}.

\begin{figure}[tp]
\centering
\subfigure[Current absorbed by the first processor]{
      \includegraphics[width=0.40\textwidth]{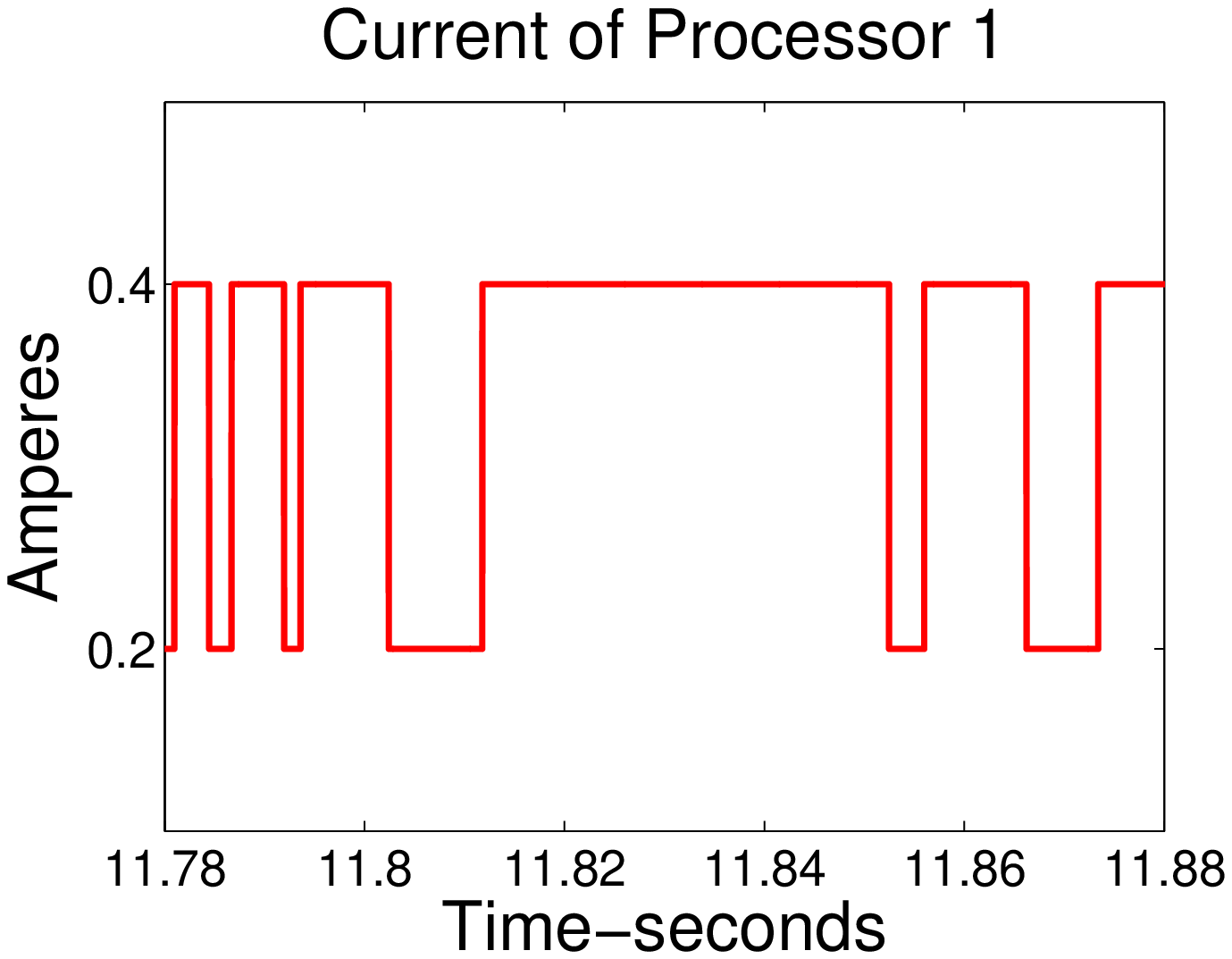}
      \label{fig:ip1}
}
\subfigure[Total current supplied by battery]{
      \includegraphics[width=0.40\textwidth]{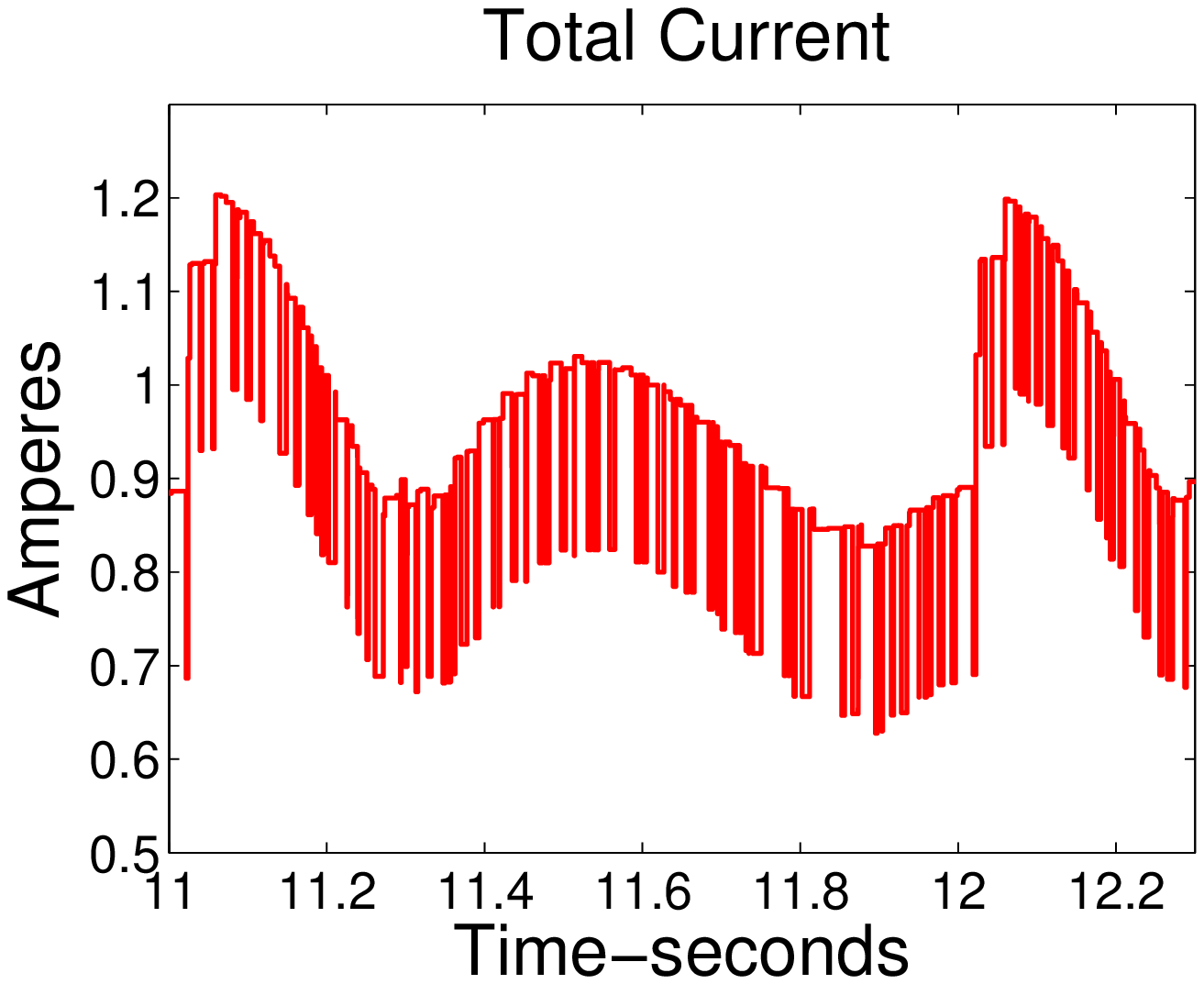}
      \label{totload}
}
\caption{Current supplied by the battery}
\end{figure}

\subsection{Robustness of real-time scheduling}\label{sec:timeuncertainty}

We demonstrate that the scheduling algorithm with a higher $B_R$ is more robust to the perturbations. Given the task set for the three pendulums with
\begin{align}
&[C^{\rm nom}_1(t), C^{\rm nom}_2(t), C^{\rm nom}_3(t)]=[4, 4, 4]{\rm ms} \cr
&[T^{\rm nom}_1(t), T^{\rm nom}_2(t), T^{\rm nom}_3(t)]=[15.4, 20.8, 30.3]{\rm ms},
\end{align}
Consider two different scheduling algorithms as the RMS algorithm and the EDF algorithm. When the tasks are scheduled under the RMS algorithm, we calculate the value of $B_R$ within $[10, 13]$s to be $8.8$ according to Definition \ref{df:BR}. When the tasks are scheduled under the EDF algorithm, we calculate the value of $B_R$ within $[10, 13]$s to be $11.4$. Since the system using the EDF algorithm has a higher measure of robustness as compared with the system using the RMS algorithm,  we conclude that the former is more robust to the perturbations considered. Indeed, under the same perturbation $\mathcal{E}(t)$, our dynamic schedulability test has confirmed that the real-time task set under the EDF algorithm is still schedulable, but is not schedulable under the RMS  algorithm.

\subsection{Robustness of the battery switching strategies}
 We compare the results from the three battery switching algorithms: Voltage Thresholding (VT),  Capacity Thresholding (CT), and Adaptive Thresholding (AT). We perform two tests comparing the behaviors of the three battery switching algorithms.

{\it Test 1:} We assume that the battery supplies the controller and the three pendulums.  Unexpected perturbations in load currents happen due to the loss of schedulability in the control tasks caused by the unexpected perturbation $\mathcal{E}(t)$ that makes  certain pendulums fail to receive updated control signals for a short period of time. To regain control a large motor current needs to be supplied, thus causing a sudden drop in the terminal voltage of the battery.

{\it Test 2:} We assume that the battery supplies different constant loads for an entire cycle (charge-discharge) of operation as the
SoC of the
battery varies. Such a test allows us to test the performance of the battery switching algorithms when dealing with a battery subjected to smooth loads of varying magnitude.

For each battery switching algorithm used in a particular test,  we simulate ten charge-discharge cycles on a 275mAh battery. After each cycle we assume that a certain amount of capacity loss occurs i.e. the value of $f_2$ decreases. We assume $f_2$ takes the values $[1,0.9,0.8,\cdots,0.1]$ over the ten cycles.

For VT we set the following criteria. A successful failure detection occurs when the terminal voltage $V\leq3.5$ volts and the estimated
SoC $\hat{x}_1\leq10\%$. A false alarm occurs if the voltage $V\leq3.5$ volts when $\hat{x}_1>10\%$. The false alarm happens when the algorithm attempts to switch out the battery on observing a temporary disturbance in load current even though the value of SoC is still larger than $10\%$.

For CT the following criteria are used. A false alarm occurs when $\hat{x}_1\leq 10\%$ and $V>3.6$ volts. This indicates that the algorithm is switching a battery out due to a perceived drop in the
SoC although the terminal voltage is approximately $2.8\%$ higher than the voltage threshold used in the previous test. The algorithm misses a fault if $\hat{x}_1\leq 10\%$ and the battery terminal voltage has fallen by $33\%$ or more from its initial no load value when $f_2=1$ and $x_1=1$.

For AT we use criteria similar to CT. A false alarm is recorded if the terminal voltage of the battery at the instant of switching is higher than $3.6$ volts. The algorithm misses a fault if the battery terminal voltage at the switching time instant has fallen by $33\%$ or more from its initial no load value when $f_2=1$ and $x_1=1$.

The test results are shown in the tables of Figure \ref{tb1}.
The total number of simulation runs per test are $T=10$. Let $H$, $F$ and $M$ be the number of successfully detected  faults, false alarms, and missed detections respectively. Note that $T=H+F+M$.  The fault detection rate (DR), false alarm rate (FAR) and the missed detection rate (MDR) are defined as $H/T$, $F/T$ and $M/T$ respectively, and $DR+FAR+MDR=1$.

\begin{figure}[tp]
\centering
\subfigure[Test 1 - results]{
  \begin{tabular}{ |c || c | c | c|} \hline \label{test1results}
    \textbf{Algorithm type} & \textbf{DR} & \textbf{FAR} & \textbf{MDR}\\ \hline \hline
    \textbf{VT} & 40\% & 60\% & 0\% \\ \hline
    \textbf{CT} & 100\% & 0\% & 0\% \\ \hline
    \textbf{AT} & 100\% & 0\% & 0\% \\
    \hline
  \end{tabular}
}
\subfigure[Test 2 - results]{
  \begin{tabular}{ ||c | c | c|} \hline \label{test2results}
    \textbf{DR} & \textbf{FAR} & \textbf{MDR}\\ \hline \hline
    50\% & 50\% & 0\% \\ \hline
    70\% & 30\% & 0\% \\ \hline
    100\% & 0\% & 0\% \\
    \hline
  \end{tabular}
}
\caption[Caption for subtables]{Battery switching algorithm test results}\label{tb1}
\end{figure}

It appears that none of the algorithms miss a fault, i.e. all of them ultimately disconnect a dying battery out of service before the terminal voltage falls below the criteria we set. VT produces false alarms six out of ten times in the presence of disturbances as shown in Figure \ref{tb1}(a). Even for smooth loads,VT produces five false alarms in ten trials as a result of changes in $f_2$ as shown in \ref{tb1}(b). It appears that CT performs well in the presence of disturbances as it produces no false alarms, however it produces three false alarms in ten trials when $f_2$ changes.  AT produces no false alarms in any case. It out-performs VT and CT in these tests.

\section{Conclusions} \label{con1}
This paper follows an analytical approach to establish  notions of robustness for real-time task scheduling algorithms and battery management algorithms. Combined with existing analytical results for robustness of control systems, our results provide a unified theoretical foundation for robustness of CPSb measured by the maximum tolerable perturbations in timing and battery capacity. Our results allow the entire system  to be analyzed using the dynamic schedulability test, battery stability test and the stability test for feedback controllers.

\bibliographystyle{acm}
\bibliography{Outline}

\begin{thebibliography}{10}

\bibitem{Abdelzaher04}
{\sc Abdelzaher, T., Sharma, V., and Lu, C.}
\newblock A utilization bound for aperiodic tasks and priority driven
  scheduling.
\newblock {\em IEEE Transactions on Computers 53}, 3 (March 2004), 334--350.

\bibitem{AbuSharkh}
{\sc Abu-Sharkh, S., and Doerffel, D.}
\newblock Rapid test and non-linear model characterisation of solid-state
  lithium-ion batteries.
\newblock {\em Journal of Power Sources 130}, 1-2 (2004), 266 -- 274.

\bibitem{Andersson_rtas}
{\sc Andersson, B., and Ekelin, C.}
\newblock Exact admission-control for integrated aperiodic and periodic tasks.
\newblock In {\em Proceedings of the 11th IEEE Symposium on Real-Time and
  Embedded Technology and Applications\/} (San Francisco, CA, USA, 2005), IEEE
  Computer Society, pp.~76--85.

\bibitem{Audsley_RTA}
{\sc Audsley, N.~C., Burns, A., Richardson, M., and Wellings, A.~J.}
\newblock Hard real-time scheduling: {T}he deadline monotonic approach.
\newblock In {\em Proc. 8th IEEE Workshop on Real-Time Operating Systems and
  Software\/} (Atlanta, GA, USA, 1991), pp.~127--132.

\bibitem{Barbarisi06}
{\sc Barbarisi, O., Vasca, F., and Glielmo, L.}
\newblock State of charge {K}alman filter estimator for automotive batteries.
\newblock {\em Control Engineering Practice 14\/} (2006), 267--275.

\bibitem{ABurns_offset}
{\sc Bate, I., and Burns, A.}
\newblock Schedulability analysis of fixed priority real-time systems with
  offsets.
\newblock In {\em 9th Euromicro Workshop on Real-Time Systems\/} (1997).

\bibitem{Iainone}
{\sc Bate, I., and Emberson, P.}
\newblock Incorporating scenarios and heuristics to improve flexibility in
  real-time embedded systems.
\newblock In {\em Proceedings of the 12th IEEE Real-Time and Embedded
  Technology and Applications Symposium\/} (2006), pp.~221--230.

\bibitem{Bini03}
{\sc Bini, E., Buttazzo, G.~C., and Buttazzo, G.~M.}
\newblock Rate monotonic analysis: {T}he hyperbolic bound.
\newblock {\em IEEE Transactions on Computers 52}, 7 (2003), 933--42.

\bibitem{Brogan}
{\sc Brogan, W.~L.}
\newblock {\em Modern Control Theory}.
\newblock Prentice Hall, October 1990.

\bibitem{Buttazzo02}
{\sc Buttazzo, G.~C., Lipari, G., Caccamo, M., and Abeni, L.}
\newblock Elastic scheduling for flexible workload management.
\newblock {\em IEEE Transactions on Computers 51}, 3 (2002), 289--302.

\bibitem{truetime}
{\sc Cervin, A., Henriksson, D., Lincoln, B., Eker, J., and Arzen, K.}
\newblock How does control timing affect performance? {A}nalysis and simulation
  of timing using {J}itterbug and {T}rue{T}ime.
\newblock {\em IEEE Control Systems Magazine 23}, 6 (June 2003), 16--30.

\bibitem{HuLemmon06}
{\sc Chantem, T., Hu, X.~S., and Lemmon, M.}
\newblock Generalized elastic scheduling.
\newblock In {\em Proc. 27th IEEE Real-Time Systems Symposium\/} (2006).

\bibitem{dis}
{\sc Chen, C.-T.}
\newblock {\em Linear System Theory and Design}, 3rd~ed.
\newblock Oxford University Press, 1998.

\bibitem{rincon}
{\sc Chen, M., and Mora, R.}
\newblock Accurate electrical battery model capable of predicting runtime and
  {I}-{V} performance.
\newblock {\em IEEE Transcations on Energy Conversion 21}, 2 (June 2006),
  504--512.

\bibitem{Coleman}
{\sc Coleman, M., Hurley, W., and Lee, C.~K.}
\newblock An improved battery characterization method using a two-pulse load
  test.
\newblock {\em IEEE Transactions on Energy Conversion 23}, 2 (June 2008), 708
  --713.

\bibitem{Iaintwo}
{\sc Emberson, P., and Bate, I.}
\newblock Minimising task migration and priority changes in mode transitions.
\newblock In {\em Proceedings of the 13th IEEE Real-Time and Embedded
  Technology and Applications Symposium\/} (2007), pp.~158--167.

\bibitem{Pandya_RTA}
{\sc Joseph, M., and Pandya, P.}
\newblock Finding response time in a real-time system.
\newblock {\em BCS Computer Journal 29}, 5 (1986), 390--395.

\bibitem{Khalil}
{\sc Khalil, H.}
\newblock {\em Nonlinear Systems}, 3rd~ed.
\newblock Prentice Hall, 2001.

\bibitem{kim_rtas}
{\sc Kim, H., and Shin, K.~G.}
\newblock On dynamic reconfiguration of a large-scale battery system.
\newblock In {\em Proceedings of the 15th IEEE Symposium on Real-Time and
  Embedded Technology and Applications\/} (Washington, DC, USA, 2009), IEEE
  Computer Society, pp.~87--96.

\bibitem{Knauff}
{\sc Knauff, M., Dafis, C., Niebur, D., Kwatny, H., and Nwankpa, C.}
\newblock Simulink model for hybrid power system test-bed.
\newblock In {\em IEEE Electric Ship Technologies Symposium, 2007\/} (May
  2007), pp.~421 --427.

\bibitem{krstic}
{\sc Krsti\'{c}, M., Kanellakopoulos, I., and Kokotovi\'{c}, P.}
\newblock {\em Nonlinear and Adaptive Control Design}.
\newblock Wiley-Interscience, 1995.

\bibitem{AKMok_1991}
{\sc Kuo, T.-W., and Mok, A.~K.}
\newblock Load adjustment in adaptive real-time systems.
\newblock In {\em Proc. 12th IEEE Real-Time Systems Symposium\/} (1991).

\bibitem{LeeCPS08}
{\sc Lee, E.~A.}
\newblock Cyber-physical systems: Design challenges.
\newblock In {\em Proceedings of the 11th IEEE Symposium on Object Oriented
  Real-Time Distributed Computing\/} (Washington, DC, USA, 2008), IEEE Computer
  Society, pp.~363--369.

\bibitem{Lehoczky_arbitrarydeadline}
{\sc Lehoczky, J.~P.}
\newblock Fixed priority scheduling of periodic task sets with arbitrary
  deadlines.
\newblock In {\em Proc. 11th IEEE Real-Time Systems Symposium\/} (Dec, 1990),
  pp.~201--209.

\bibitem{Lehoczky_1989}
{\sc Lehoczky, J.~P., Sha, L., and Ding, D.~Y.}
\newblock The rate monotonice scheduling algorithm: {E}xact characterization
  and average case behavior.
\newblock In {\em Proc. 10th IEEE Real-Time Systems Symposium\/} (1989),
  pp.~166--171.

\bibitem{BattHandBook}
{\sc Linden, D., and Reddy, T.}
\newblock {\em Handbook of Batteries}, 3rd edition~ed.
\newblock McGraw-Hill, 2002.

\bibitem{liulayland73}
{\sc Liu, C., and Layland, J.~W.}
\newblock Scheduling alghorithms for multiprogramming in a hard real-time
  environment.
\newblock {\em Journal of the Association for Computing Machineray 20}, 1
  (January 1973), 46 --61.

\bibitem{philipsbook}
{\sc Pop, V., Bergveld, H., Danilov, D., and Regtien, P.}
\newblock {\em Battery Management Systems: Accurate State-of-Charge Indication
  for Battery Powered Applications}.
\newblock Springer, 2008.

\bibitem{Rakhmatov2}
{\sc Rakhmatov, D., and Vrudhula, S.}
\newblock Energy management for battery-powered embedded systems.
\newblock {\em ACM Transactions on Embedded Computing Systems 2}, 3 (2003),
  277--324.

\bibitem{rvw_vlsi}
{\sc Rakhmatov, D., Vrudhula, S., and Wallach, D.~A.}
\newblock A model for battery lifetime analysis for organizing applications on
  a pocket computer.
\newblock {\em IEEE Transactions on VLSI Systems 11}, 6 (2003), 1019--1030.

\bibitem{rvw_model}
{\sc Rao, R., Vrudhula, S., and Rakhmatov, D.~N.}
\newblock Battery modeling for energy-aware system design.
\newblock {\em Computer 36}, 12 (2003), 77--87.

\bibitem{Regehr_rtss02}
{\sc Regehr, J.}
\newblock Scheduling tasks with mixed preemption relations for robustness to
  timing faults.
\newblock In {\em Proceedings of the 23rd IEEE Real-Time Systems Symposium\/}
  (Austin, CA, USA, 2002), IEEE Computer Society, pp.~325--326.

\bibitem{Stankovic_MSFreport}
{\sc Regehr, J., Jones, M.~B., and Stankovic, J.~A.}
\newblock Operating systems support for multimedia: The programming model
  matters.
\newblock Tech. rep., Microsoft Research Technical Report MSR-TR- 2000-89,
  September 2000.

\bibitem{machines}
{\sc Sarma, M.}
\newblock {\em Electric Machines: Steady-State Theory and Dynamic Performance}.
\newblock CL-Engineering, 1997.

\bibitem{Bernhardt}
{\sc Schweighofer, B., Raab, K., and Brasseur, G.}
\newblock Modeling of high power automotive batteries by the use of an
  automated test system.
\newblock {\em IEEE Transactions on Instrumentation and Measurement 52}, 4
  (Aug. 2003), 1087 -- 1091.

\bibitem{ShaCPS09}
{\sc Sha, L., Gopalakrishnan, S., Liu, X., and Wang, Q.}
\newblock Cyber-physical systems: A new frontier.
\newblock In {\em Machine Learning in Cyber Trust Security, Privacy, and
  Reliability}. Springer, 2009, pp.~3--13.

\bibitem{WolfCPS09}
{\sc Wolf, W.}
\newblock Cyber-physical systems.
\newblock {\em Computers 42}, 3 (2009), 88 -- 9.

\bibitem{Fenxiang}
{\sc Zhang, F., and Burns, A.}
\newblock Schedulability analysis for real-time systems with {EDF} scheduling.
\newblock {\em IEEE Transactions on Computers 58}, 9 (September 2009),
  1250--1258.

\bibitem{ZWS_RTSS08}
{\sc Zhang, F., Szwaykowska, K., Mooney, V., and Wolf, W.}
\newblock Task scheduling for control oriented requirements for cyber-physical
  systems.
\newblock In {\em Proc. of 29th IEEE Real-Time Systems Symposium\/} (Barcelona,
  Spain, 2008), pp.~47--56.

\bibitem{kz}
{\sc Zhou, K., and Doyle, J.}
\newblock {\em Essentials of Robust Control}.
\newblock Prentice Hall, 1997.

\end{thebibliography}

\end{document}